\documentclass[12pt,final]{colt2019} 
\usepackage{url,ifthen}
\usepackage{srcltx}
\usepackage{multirow}
\usepackage{boxedminipage}

\usepackage{nicefrac}
\usepackage{xspace}
\usepackage{graphicx}
\usepackage{srcltx}
\usepackage{verbatim}
\usepackage{xspace}
\definecolor{DarkGreen}{rgb}{0.1,0.5,0.1}
\definecolor{DarkRed}{rgb}{0.5,0.1,0.1}
\definecolor{DarkBlue}{rgb}{0.1,0.1,0.5}
\usepackage[small]{caption}
\usepackage{float}
\usepackage{balance}
\usepackage{amsmath}
\usepackage{amsfonts}
\usepackage{amssymb}
\usepackage{bbm}
\usepackage{dsfont}
\usepackage{todonotes}
\usepackage[utf8]{inputenc} 
\usepackage[T1]{fontenc}

\newtheorem*{namedtheorem}{\theoremname}
\newcommand{\theoremname}{testing}

\newtheorem{lem}[theorem]{Lemma}
\newtheorem{prop}[theorem]{Proposition}

\newtheorem*{question*}{Question}

\newtheorem{Alg}{Algorithm}

\definecolor{DarkGreen}{rgb}{0.1,0.5,0.1}
\definecolor{DarkRed}{rgb}{0.5,0.1,0.1}
\definecolor{DarkBlue}{rgb}{0.1,0.1,0.5}

\newcommand{\ignore}[1]{}

\renewcommand{\Pr}{\mathop{\bf Pr\/}}                    
\newcommand{\E}{\mathop{\bf E\/}}

\newcommand{\Var}{\mathop{\bf Var\/}}

\newcommand{\KL}{\mathop{\bf KL\/}}
\newcommand{\SKL}{\mathop{\bf SKL\/}}
\newcommand{\TV}{\mathop{\bf TV\/}}
\newcommand{\Hel}{\mathop{\bf H\/}}
\newcommand{\Ber}{\mathop{\bf Ber\/}}
\newcommand{\Proj}{\operatorname{Proj}}

\DeclareMathOperator{\sgn}{sgn}

\newcommand{\R}{\mathbb R}

\newcommand{\eps}{\varepsilon}

\newcommand{\fnote}[1]{\textcolor{orange}{[Fred: #1]}}
\newcommand{\vnote}[1]{\textcolor{purple}{[Vishesh: #1]}}

\allowdisplaybreaks

\makeatletter
\floatstyle{ruled}
\newfloat{fragment}{H}{lop}
\floatname{fragment}{Game}
\renewcommand{\floatc@ruled}[2]{\vspace{2pt}{\@fs@cfont \#1.\:} \#2 \par
 \vspace{1pt}}
\makeatother

\begin{document}
\title{Accuracy-Memory Tradeoffs and Phase Transitions in~ \\ Belief Propagation}
\coltauthor{%
\Name{Vishesh Jain}
\Email{visheshj@mit.edu}\\
\Name{Frederic Koehler} \Email{fkoehler@mit.edu}\\
\addr Massachusetts Institute of Technology. Department of Mathematics.
\AND 
\Name{Jingbo Liu} \Email{jingbo@mit.edu} \\
\addr Massachusetts Institute of Technology. IDSS.
\AND
\Name{Elchanan Mossel} \Email{elmos@mit.edu} \\
\addr Massachusetts Institute of Technology. Department of Mathematics and IDSS.
}

\makeatletter
\def\blfootnote{\xdef\@thefnmark{}\@footnotetext}
\makeatother

\maketitle
\begin{abstract}
\blfootnote{Authors are sorted alphabetically.}
    The analysis of Belief Propagation and other algorithms for the {\em reconstruction problem} plays a key role in the analysis of community detection in inference on graphs, phylogenetic reconstruction in bioinformatics, and the cavity method in statistical physics. 
    We prove a conjecture of Evans, Kenyon, Peres, and Schulman (2000) which states that any bounded memory message passing algorithm is statistically much weaker than Belief Propagation for the reconstruction problem. More formally, any recursive algorithm with bounded memory for the  reconstruction problem on the trees with the binary symmetric channel has a phase transition strictly below the Belief Propagation threshold, also known as the Kesten-Stigum bound. The proof combines in novel fashion tools from recursive reconstruction, information theory, and optimal transport, and also establishes an asymptotic normality result for BP and other message-passing algorithms near the critical threshold.

\end{abstract}

\section{Introduction}
Belief Propagation is one of the most popular algorithms in graphical models. The main result of this paper (Theorem~\ref{thm:main}) shows that bounded memory variants of Belief Propagation have no asymptotic statistical power in regimes where Belief Propagation does. This proves a long-standing conjecture (Conjecture~\ref{conj:peres})  \citep{evans2000broadcasting}.

    \vspace{.1cm}
\textbf{Belief Propagation: }Belief Propagation (BP) is one of most popular algorithms in machine learning and probabilistic inference \citep{Pearl88}. 
It is also a key algorithm and a key analytic tool in statistical physics with applications to inference problems and coding 
where it is an important ingredient of replica analysis (e.g.~\cite{mezard-montanari}),
and in probability theory, where it is studied under the names of ``broadcasting on trees'' and the ``reconstruction problem on trees'' (e.g.~\cite{mossel2004survey}). 
The analysis of BP on trees plays a crucial role in many inference problems arising in different fields. In the problem of phylogenetic reconstruction arising in biology, both belief propagation and bounded memory algorithms like recursive majority have been extensively studied in theory and practice, and they have also played an important role in works on learning phylogenies (the underlying tree structure): see  e.g.~\cite{mossel2004phase,daskalakis2006optimal}.
In the analysis of community detection in block models, BP on trees and related message-passing algorithms (e.g. linearizations of BP) play a fundamental role in predicting and rigorously analyzing the recoverability of community structure, as the sparse SBM is locally tree-like: see e.g. \cite{decelle2011inference,krzakala2013spectral,mossel2015reconstruction,mossel2014belief}.
Finally, there is a long history of BP being studied and used in the theory of error correcting codes (e.g.~\cite{RichardsonUrbanke:01,montanari2005tight}).

One of the reasons for the popularity of BP is the fact that its time complexity is linear in the number of nodes in the factor graph  (assuming real-number operations count as one operation),
while a brute force algorithm generally has exponential complexity. 
Given that the algorithm is a simple recursive algorithm that is easy to implement and runs in linear time, it is natural to ask how fragile it is. 
In particular, are there bounded (bit) memory variants of the algorithm that are as statistically efficient as the algorithm itself? Is the algorithm robust to a small amount of noise during it execution? These natural problems, which were open for almost two decades, intimately relate to a recent impressive body of work in machine learning which tries to understand the statistical implications of computationally limited algorithms.

    \vspace{.1cm}
\textbf{Statistical efficiency of computationally efficient algorithms: }Understanding the power of computationally limited algorithms for inference tasks is a major (or perhaps the) task of computational learning theory. A recent trend in this area deals with reductions between inference problems on graphs that are known to be informational theoretically solvable but are assumed to be unsolvable in polynomial time. Thus a recent line of work including \cite{berthet2013optimal,ma2015computational,brennan2018reducibility}
aims to prove that for certain inference problems a \emph{computational-statistical gap} exists: more data is needed to infer in polynomial time than is needed information theoretically. For many other problems, no reductions are known though it is believed that similar phenomena occur. 
Some notable recent examples include the multi-community stochastic block model \citep{decelle2011inference,zdeborova2016statistical} and sparse linear regression (see e.g. \cite{zhang2017optimal}). 
Interestingly, in most of the problems discussed above, BP and other message passing algorithms are either known or conjectured to be the optimal algorithms among all computationally efficient algorithms. 

Can statements proving computational-statistical gaps be proved unconditionally (i.e. not by reduction)? Very impressive results were recently proven by \cite{raz2018fast} and follow up work \citep{raz2017time,kol2017time,garg2018extractor,moshkovitz2017mixing} for learning tasks with bounded memory, where it is shown that unless the memory is quadratic in the instance size, the running time of the algorithm has to be exponential.

    \vspace{.1cm}
{\bf Communication complexity of distributed estimation: }In many problems concerning the trade-offs between the communication complexity and the statistical risk (or other system performance measures),
a useful tool for establishing lower bounds is the \emph{strong data processing inequality},
which dictates how fast the mutual information must decay along a Markov chain (\cite{ahlswede1986}; \cite{zhang2013}; 
\cite{liu2014key};
\cite{lcv2015};
\cite{braverman2016communication};
\cite{liu2017secret};
\cite{xu2017information}; \cite{hlps19}).
Yet, sometimes, strictly better lower bounds may be obtained by replacing the strong data processing argument with a careful analysis of the contraction of the Fisher information or $\chi^2$-information due to compression (\cite{han2018geometric}; \cite{barnes2018geometric}; \cite{acharya2018inference}).
For example, \cite{barnes2018geometric}  proved the contraction of the Fisher information in the Gaussian location model via a geometric analysis of the quantization of the score function (which is a high dimensional Gaussian vector).
The $\chi^2$-contraction idea is relevant to the BP with bounded memory problem considered in the current paper,
but a key difficulty (which we resolved using novel optimal transportation methods) is to show a Gaussian approximation result for ``good'' reconstruction algorithms.

    \vspace{.1cm}
\textbf{Our results: }Our results show that any message-passing algorithm using finite memory is much weaker than BP in the following sense: 
there is a range of parameters of the model for which BP is a good estimator while any bounded message-passing algorithm has no statistical power. This has immediate implications for applications of BP in phylogeny, for the block model and for population dynamics, as this imply that in these applications too, the algorithms used (BP or others) cannot be replaced by bounded memory message-passing algorithms. 
We proceed with definitions and formal statements of the main results. 

\subsection{Broadcasting on trees}

On a rooted tree $T$ with root $\rho$, the \emph{broadcast process} with error probability $\varepsilon \in (0,1/2)$ is defined as follows, generating labels $X_v \in \{\pm 1\}$ for every vertex $v \in T$. The label $X_{\rho}$ of the root $\rho$ is assigned either $+1$ or $-1$ with equal probability, and then for each edge $e=(v_1,v_2)$, the probability that the labels assigned to $v_1$ and $v_2$ are different is $\varepsilon$, independent of all other edges. This model has several interpretations. In communication theory, one may consider each of the edges as an independent binary symmetric channel. In biology, one may say that there is some binary property each child inherits from its parent independently for all children. This model is also an example of an Ising model on $T$ with constant interaction strength \cite{lyons1989ising}. 
We refer the reader to \cite{evans2000broadcasting} and \cite{ mossel2004survey} for a detailed account of the history of this model, as well as classical and modern results. 

One of the most fundamental questions about this process is the following: for a given set of vertices $V$ (e.g. the leaves of a finite tree), can we typically infer the original label correctly from the labels at $V$? More precisely, let $p(T, V,\varepsilon)$ denote the probability of reconstructing the label at the root given the labels at $V$. Since random guessing succeeds in reconstructing the original label with probability $1/2$, it is natural to write $p(T,V,\varepsilon):= (1/2) + s(T,V,\varepsilon)$, so that $s(T,V,\varepsilon)$ represents the advantage over random guessing gained by having access to the labels at $V$. 
For an infinite tree $T$, a natural question to ask is whether there is a \emph{uniform} advantage over random guessing provided by knowing the labels at the vertices at \emph{any} level of the tree. Formally, we let $V_n$ denote the set of all vertices at distance $n$ from the root $\rho$ 
and ask whether $\lim_{n\to \infty}s(T,V_n,\varepsilon) = \inf_{n\geq 1}s(T,V_n,\varepsilon) > 0$ (the first equality follows by the data-processing inequality); 
if so, we say that the \emph{reconstruction problem} for $T$ at $\varepsilon$ is \emph{solvable}. The solvability threshold is determined by the branching number of $T$:
\begin{theorem}[\cite{evans2000broadcasting}]
\label{thm:EKPS}
Consider the broadcasting process with parameter $\varepsilon$ on an infinite tree $T$ with branching number $\text{br}(T)$. Then, with $\varepsilon_{c} := (1-\text{br}{T}^{-1/2})/2$,
\begin{equation}
\label{eqn:threshold}
\lim_{n\to \infty}s(T,V_n,\varepsilon) 
\begin{cases}
>0, \quad \text{if } \varepsilon < \varepsilon_{c} \\
=0, \quad \text{if } \varepsilon > \varepsilon_{c}. 
\end{cases}
\end{equation}
(Note that for $d$-ary trees, $\text{br}(T) = d$. See Definition~\ref{defn:branching-number} for the general definition.)
\end{theorem}

Remarkably, the exact same threshold also dictates the limit of \emph{weak recovery} for the 2-community stochastic block model \citep{mossel2015reconstruction,Massoulie:13,mossel2018proof}. Intuitively, this is because locally around a typical node, the SBM graph looks like a Galton-Watson tree, and the community assignment of nodes can be coupled to the aforementioned broadcast process on the tree.
The same threshold also dictates the phase transition for the phylogenetic reconstruction problem, where the goal is to reconstruct the underlying tree \citep{mossel2004phase, daskalakis2006optimal, Steel:01}. The intuition here is that information about the deep part of the structure of the tree is transmitted via the broadcast channel. 

 The proof that the above limit is positive for $\varepsilon < \varepsilon_c$ first appears in \cite{KestenStigum:66} 
 The proof that the limit is $0$ when $\varepsilon > \varepsilon_c$ is harder, and only appeared around two decades later, first for regular trees (in which case $\text{br}(T)$ coincides with the arity of $T$) in \cite{BlRuZa:95}, and then for general trees in \cite{evans2000broadcasting}. Subsequently, different proofs appeared in \cite{Ioffe;96b} and \cite{berger2005glauber}. 

\subsection{Message-passing algorithms for reconstruction}
For simplicity of notation, we focus on the setting where $V$, the set of revealed nodes, are the leaves of some finite-depth tree.
Then a \emph{message-passing} algorithm for reconstructing the label at the root, given the labels $X_V$ at the set of leaves $V$ of the tree
is specified by the following data: (i) a \emph{message space} $\Sigma$ (possibly infinite) to which messages belong; 
(ii) initial messages $(Y_v)_{v \in V}$ in $\Sigma^{V}$, where $Y_v$ is a (possibly random) function of $X_v$ for all ${v \in V}$; and (iii) for each vertex $u \in T$, a fixed \emph{reconstruction function} $f_{u}:\Sigma^{C(u)}\to \Sigma$ 
, where $C(u)$ denotes the children of $u$ and $f_u$ is allowed to be a randomized function (i.e. a channel;  
The randomness of this channel is independent of the randomness of the tree broadcast process).
Reconstruction proceeds recursively -- the message output by node $u$ (to its parent) is
\begin{equation} Y_u = f_u(Y_{C(u)}). \label{eqn:y_u}\end{equation}
To visualize, we can think of the $X$ as living on a ``broadcasting tree'' and the $Y$ as living on a mirrored ``reconstruction tree'' (see Figure~\ref{fig_1}). Letting $Y_{\pm}$ denote the random variables corresponding to the output $Y_{\rho}$ of $f_\rho$ under $T_V^{\pm}$ (the distribution of labels at the leaves $V$, conditioned on the label at $\rho$ being $\pm$), the probability of correct reconstruction of the root given $Y_{\rho}$ is $0.5(1+\TV(Y_+,Y_-))$. Thus, a natural measure of the power of the message passing algorithm is 
$\TV(Y_+,Y_-)$.
The size of the message space plays a crucial role in this paper;
we call $\log_2 |\Sigma|$ the number of bits of memory used by the algorithm.

In this context BP is just the usual recursive scheme for computing exactly the marginal distribution of the label at $\rho$, given the labels at $V$. Explicitly, $\Sigma = [-1,1]$, 
$Y_v = X_v$ for leaf nodes $v$, and the reconstruction function at every node $u$ is (by Bayes rule, with $\theta = 1-2\epsilon$)
\[ f_{u}^{(BP)}\left(y_1,\ldots,y_{|C(u)|}\right) := \frac{\prod_i(1 + \theta y_i) - \prod_i(1 - \theta y_i)}{\prod_i(1 + \theta y_i) + \prod_i(1 - \theta y_i)} \]
Note that, by definition, outputting the more likely label under the marginal is the Bayes optimal classifier in this setting i.e. it achieves the maximum probability of reconstruction among \emph{all} algorithms. In particular, the advantage of belief propagation over random guessing enjoys the limiting behavior in \eqref{eqn:threshold}. 

\subsection{Limitations of bounded memory algorithms}
Another natural algorithm for the reconstruction problem is to output the label present at a majority of the vertices in $V$. While this does not achieve the same probability of success as belief propagation, it is known that it \emph{does} achieve the limiting behavior in \eqref{eqn:threshold} \citep{KestenStigum:67}. 
Note that this is a message-passing algorithm with $\Sigma = \mathbb{Z}$ and $f$ which sums its inputs.

What if $\Sigma$ is a bounded size alphabet? In the case $|\Sigma| = 2$, the natural message-passing variant of this algorithm is to estimate the label at $\rho$ via \emph{recursive majority} i.e. $\Sigma=\{\pm 1\}$, and $f_{u}\colon \Sigma^{C(u)}\to \Sigma$ is the majority function (we assume for convenience that $|C(u)|$ is odd for each $u$). Note that recursive majority is easier to implement than belief propagation, not requiring access to $\varepsilon$; for this reason among others, it is quite popular in practice, in particular in biological applications. 

The following striking conjecture states that reconstruction down to the KS threshold requires unbounded memory, i.e. there is no bounded-memory analogue of BP:  
\begin{conjecture}[\cite{evans2000broadcasting}]\label{conj:peres}
For any fixed $L > 0$, no message-passing algorithm on an alphabet of size $L$ can achieve the guarantee of \eqref{eqn:threshold}. In other words, there exists a fixed noise level $\varepsilon(L)$ such that reconstruction is information-theoretically possible but no such message-passing algorithm is asymptotically better than a random guess.
\end{conjecture}
The conjecture is also discussed in~\citep{mossel2004survey}. 
\cite{mossel1998recursive} verified this conjecture on periodic trees in the special case when $\Sigma = \{\pm 1\}$, the reconstruction function $f_u = f$ is the same for all nodes $u$, and $f(-x) = -f(x)$ for all inputs $x$; in particular, this includes the important case of recursive majority. 

\emph{The main result of this paper is to verify Conjecture~\ref{conj:peres} on the $d$-ary tree\footnote{We will also assume for convenience that $f_u$ is constant at each level of the tree.} for all $L > 0$.} The proof will rely upon a careful analysis of the distributional recursion induced by combining the recursive broadcast and reconstruction (message-passing) steps. In fact, we present two approaches along these lines: a relatively elementary approach which proves the result for one bit memory  ($L = 2$) and illustrates many of the main difficulties in this problem, and a higher-powered method (which crucially builds upon Wasserstein estimates from optimal transport theory) which proves the result for all $L$ and even pins down the correct quantitative dependence on the distance to the threshold:
\begin{theorem}
\label{thm:main}
For any integer $d\ge 2$, there exist positive real numbers $c_d$ and $C_d$ such that the following holds:
for any fixed $L$, 
the maximum error probability $\varepsilon(L)$ for which there exists a message-passing algorithm with alphabet size $L$ guaranteeing asymptotic reconstruction on the infinite $d$-ary tree satisfies
\begin{align}
    L^{-C_d} \le \varepsilon_c - \varepsilon(L) \le L^{-c_d}.
\end{align}
\end{theorem}
As a by product, we remark that the (renormalized) BP message distribution at the root of the infinite $d$-ary tree approaches a Gaussian as $\varepsilon$ approaches criticality. 
More precisely:
\begin{corollary}\label{corr:bp-fixpoint} 
Fix $d \ge 2$ and broadcasting parameter $\varepsilon \in (0,\varepsilon_c)$. Let $\rho$ denote the root of a $d$-ary tree of depth $n$, $V_n$ denote the set of leaves, and let $Y_n := \E[X_{\rho} | X_{V_n}]$ under the broadcast process. Then there exists a limit r.v. $Y$ such that $Y_n \to Y$ in distribution and
\[ W_2\left(\frac{Y}{\sqrt{\Var(Y)}}, \,G\right) \le (\varepsilon_c - \varepsilon)^{C_d}, \]
for some $C_d > 0$ independent of $\varepsilon$,
where $G \sim N(0,1)$ and $W_2$ denotes the $2$-Wasserstein distance (see e.g.\ the definition in \cite{villani2003topics}).
\end{corollary}
Let us emphasize that the lower bound in  Theorem~\ref{thm:main} applies to general reconstruction schemes (not necessarily discretized BP),
even though Corollary~\ref{corr:bp-fixpoint} only concerns the specific BP algorithm.
We also note that Gaussian approximation of BP is widely used in {\em density evolution} analysis. This is a different setup where the number of iterations is bounded but the degree goes to infinity, see e.g. \cite{bayati2011dynamics}.
Note in particular, that in the density evolution setting, normal approximation is used both above and below the reconstruction threshold. 

A subsequent work to ours, ~\cite{moitra2019circuit}, studies the complexity of Belief Propagation from the point of view of circuit complexity. 
Most relevant to us is the result of ~\cite{moitra2019circuit} showing that BP can be computed in $\mathbf{NC}^1$. Thus, 
there exist a circuit of depth $O(n)$ with binary AND and OR gates and NOT gates that computes Belief Propagation in the following sense. The input to gates are the leaves values (repeated many times). The circuit returns a bit that agrees with the more likely posterior according to BP, whenever the BP posterior has a bias of more than $1/d^n$. 
These results hold independently of broadcast parameter $\eps$. The results of ~\cite{moitra2019circuit} do not contradict the results of the current paper as the circuit constructed does not confirm to a tree topology with each input bits appearing only once. Rather, in the circuit constructed each input bit is repeated $d^{O(n)}$ times. In other results, ~\cite{moitra2019circuit} show that bounded depth circuits with AND and OR gates (the class $\mathbf{AC^0}$) cannot compute a nontrivial approximation to BP even in an average sense above the KS bound.
\textbf{Organization:} As described above, we first sketch a more elementary proof of the lower bound (impossibility result) in the $L = 2$ case of Theorem~\ref{thm:main} in Section~\ref{sec:1-bit-sketch}, then prove it completely (along with Corollary~\ref{corr:bp-fixpoint}) with a more powerful approach in Section~\ref{sec:multibit}. 
Missing proofs for the converse part are given in Appendix~\ref{sec:1-bit} and \ref{sec:appendix-multibit}.
The matching upper bound via a quantization of BP is given in Appendix~\ref{sec:achievability}.

\section{Impossibility of 1-bit reconstruction near criticality}
\label{sec:1-bit-sketch}
In this section, we will sketch the main ideas behind a relatively simple and self-contained information theoretic proof of the impossibility of 1-bit message-passing algorithms solving the reconstruction problem all the way to the Kesten-Stigum (KS) threshold. 
Our analysis of this case will also serve to illustrate the challenges encountered towards resolving Conjecture~\ref{conj:peres}, and shed additional light on its ultimate resolution in the next section. Complete statements and proofs for this section are provided in Appendix~\ref{sec:1-bit}. 


Throughout this section, we will adopt the convenient reparameterization $\varepsilon = 1/2 - \nu$. Note that with this reparameterization, the KS threshold corresponds to $4d\nu^{2} = 1$ i.e. the reconstruction problem is solvable if $4d\nu^{2} > 1$ and unsolvable if $4d\nu^{2} < 1$. 

\subsection{A direct proof of the Kesten-Stigum bound}
\label{subsec:KS-direct-proof-sketch}
Here, for simplicity, we will only discuss the KS bound in the case of the
infinite $d$-ary tree, deferring the general case to Section~\ref{sec:appendix-direct-KS-bound}. 
Denoting by $T_{n}^{\pm}$ the distributions
on the labels $(X_v)_{v \in V}$ for the leaves $V$ of the depth $n$ tree, 
conditioned on the root being $\pm$. Recall that 
$2s(T,V_{n},\varepsilon)=\TV(T_{n}^{+},T_{n}^{-})$.
Note also that by considering the labels at the vertices one level
below the root, it is easily seen that
\[
T_{n}^{\pm}\sim\left(\left(\frac{1}{2}+\nu\right)T_{n-1}^{\pm}+\left(\frac{1}{2}-\nu\right)T_{n-1}^{\mp}\right)^{\otimes d}.
\]
Owing to this recursive structure of the problem, it is much more
convenient to switch to an information measure which tensorizes well
(and which is also `stronger' than TV). Here, we make the choice of
working with the symmetrized version of KL-divergence (also known
as Jeffrey's divergence), defined by $\SKL(P,Q):=\KL(P,Q)+\KL(Q,P)$;
by Pinsker's inequality, $\SKL(T^+_{n},T^-_{n})\to0$ shows that $\TV(T^+_{n},T^-_{n})\to0$ as well.

Of key importance to us is the fact that SKL behaves very well under
`symmetric mixtures'; in Lemma~\ref{lem:mixing-inequality} we show using a short direct
computation that 
\[
\SKL\left(\left(\frac{1}{2}+\nu\right)P+\left(\frac{1}{2}-\nu\right)Q,\left(\frac{1}{2}-\nu\right)P+\left(\frac{1}{2}+\nu\right)Q\right)\leq4\nu^{2}\SKL(P,Q).
\]
Given this `mixing inequality', the proof of the KS bound
is now immediate. Indeed, 
\begin{align*}
\SKL(T_{n}^{+},T_{n}^{-}) & =d\SKL\left(\left(\frac{1}{2}+\nu\right)T_{n-1}^{+}+\left(\frac{1}{2}-\nu\right)T_{n-1}^{-},\left(\frac{1}{2}-\nu\right)T_{n-1}^{+}+\left(\frac{1}{2}+\nu\right)T_{n-1}^{-}\right)\\
 & \leq4\nu^{2}d\SKL(T_{n-1}^{+},T_{n-1}^{-})\leq(4\nu^{2}d)^{n-1}\SKL(T_{1}^{+},T_{1}^{-})=O\left((4\nu^{2}d)^{n-1}\right),
\end{align*}
so we see that $\lim_{n\to\infty}\SKL(T_{n}^{+},T_{n}^{-})=0$ if
$4\nu^{2}d<1$.

\subsection{Interlude: noisy message-passing algorithms fail near criticality}
\label{subsec:noisy-message-sketch}

Showing that `noisy' message-passing algorithms fail near criticality requires only a slight extension of the above discussion,
and is a natural segue into our discussion of 1-bit message-passing
algorithms. Here, by a noisy message-passing algorithm, we mean 
that for every node $u$ in the tree, messages from the
children of $u$ to $u$ are processed through independent copies
of a noisy channel $P_{Y|X}:\Sigma\to\Sigma$. In this setting, instead
of $T_{n}^{\pm}$, the natural choice of distributions to look at
are $P_{n}^{\pm}$, where $P_{n}^{\pm}$ denotes the distribution
on $\Sigma$ (which we interpret as the final message from
the root) obtained by broadcasting $n$ levels down, conditioned on
the root being $\pm$, and reconstructing using our (noisy) message-passing
algorithm. Once again, by considering the labels at the vertices one
level below the root, it is immediate that 
\begin{equation}\label{eqn:body-p_n-recursion}
P_{n}^{\pm}=f_{*}\left(\left(P_{Y|X}\circ\left(\left(\frac{1}{2}+\nu\right)P_{n-1}^{\pm}+\left(\frac{1}{2}-\nu\right)P_{n-1}^{\mp}\right)\right)^{\otimes d}\right),
\end{equation}
where $f:\Sigma^{d}\to\Sigma$ denotes the reconstruction function
at the root, and $f_{*}(\mu)$ denotes the pushforward of the measure
$\mu$ by the function $f$. In particular, it follows that if (we
show in Examples~\ref{example:mixture-noise}, \ref{example:bp-additive-noise} that this is indeed the case for many channels of interest)
the channel $P_{Y|X}$ satisfies a strong data-processing inequality
(SDPI) i.e. there exists some constant $\eta\in[0,1)$ such that for
all distributions $P,Q$ on $\Sigma$, 
$\SKL(P_{Y|X}\circ P,P_{Y|X}\circ Q)\leq\eta\SKL(P,Q)$
then it follows by a similar computation as
above that 
\[
\SKL(P_{n}^{+},P_{n}^{-})\leq4\nu^{2}\eta d\SKL\left(P_{n-1}^{+},P_{n-1}^{-}\right)=O\left((4\nu^{2}\eta d)^{n-1}\right),
\]
so we see that such algorithms can solve the reconstruction problem
only if $4\nu^{2}d\geq\eta^{-1}$ i.e. they do not work all the
way to the KS threshold (Theorem~\ref{thm:KS-noisy-general}).

\subsection{1-bit message-passing algorithms fail near criticality}
\label{subsec:1-bit-sketch}
Motivated by the observation that for a fixed finite alphabet $\Sigma$,
any function $f:\Sigma^{d}\to\Sigma$ cannot be injective (for $d$
large enough) on the support of any non-trivial product distribution
on $\Sigma^{d}$, and therefore, must `lose' information, it is tempting
to think that the above strategy for noisy message-passing algorithms
can be adapted directly to show that finite-bit message-passing algorithms fail near criticality as well. However, it is not
the case that a general function $f:\Sigma^{d}\to\Sigma$ satisfies
an SDPI, even when $\Sigma=\{0,1\}$:
\begin{example}[No SDPI for general $f$]\label{example:or-function}
For $P=\Ber(p),Q=\Ber(q)$, and $f:\{0,1\}^{d}\to\{0,1\}$ equal to
the OR-function, 
\[
\lim_{p,q\to0}\frac{\SKL(f_{*}(P^{\otimes d}),f_{*}(Q^{\otimes d}))}{\SKL(P^{\otimes d},Q^{\otimes d})}=1.
\]
This is because in the limit we can disregard all events
as negligible except that either all inputs are 0, or that a single
input is 1, and the OR function memorizes which event occurred.
\end{example}
On the other hand, it turns out that our original intuition is `mostly correct':
more precisely, we will show (Theorem~\ref{thm:restricted-sdpi-multibit}) that such functions do indeed satisfy
an SDPI provided
the input distributions under consideration have `robust full support'
i.e. they assume each symbol of the alphabet $\Sigma$ with
probability at least some uniform positive constant. In particular,
this shows that any potential finite-bit message-passing algorithm
which succeeds near criticality must get close to the boundary of
the probability simplex in $\R^{\Sigma}$ infinitely often. 

In the case when $\Sigma=\{0,1\}$, we can say even more, and prove
an inverse theorem (Theorem~\ref{thm:inverse-theorem}) for the non-contraction of SKL: not only can SKL
non-contraction only occur near the boundary of the probability simplex
(in this case, 
identified with $[0,1]$), but also, the functions achieving such
non-contraction are only those (Definition~\ref{defn:OR-like}) with behavior similar to the OR-function
(for $p_{n}^{+},p_{n}^{-}$ close to $0$), in that they are able to distinguish the
all 0s input from inputs with a single 1, or symmetrically for AND (with $p_{n}^{+},p_{n}^{-}$ close to $1$).
From this we see (Theorem~\ref{thm:no-reconstruction-fixed-function}) that 1-bit algorithms with the same reconstruction function at every node fail near criticality (eliminating the symmetry assumption from \cite{mossel1998recursive}), because the only mechanism which prevents contraction in SKL (behaving like
OR or AND near appropriate boundaries) also `repels' the distributional
iterates away from the boundary. 

In Theorem~\ref{thm:no-1bit-reconstruction}, we consider the 1-bit reconstruction problem when the reconstruction
functions at different levels are allowed to be different. 
This larger class includes natural reconstruction schemes such as the so-called `TRIBES' function from Boolean analysis, which uses either the AND-function
or the OR-function depending on the level: the potential
problem is that such an algorithm could alternate `losing' steps in
which the distributional dynamics go towards the boundary of the probability
simplex $[0,1]$ with `gaining' steps, where a function like AND or
OR (depending on the boundary) is applied to gain in SKL 
(note that we are assuming that $4d\nu^{2}>1$). 

To overcome
this obstacle, we introduce a Lyapunov function for our discrete time
dynamical system; more precisely, we define a function $\phi$ such that
$\phi \to-\infty$ implies that $\SKL(P^{+}_{n},P^{-}_{n})\to0$, and for which
we can show that $\phi$ decreases at every step. 
The essential idea is to
define $\phi(P,Q)$ to be $\log\SKL(P,Q)$ plus some `negative log-barrier' term,
which penalizes $\phi(P,Q)$ for moving away from
the boundary --- by carefully balancing these terms,
we can ensure $\phi$ indeed goes down at every step.
Finally, since the log-barrier term is bounded from
below, it follows that $\phi(P^{+}_{n},P^{-}_{n})\to-\infty$ implies that $\SKL(P^{+}_{n},P^{-}_{n})\to0$. 

\section{Impossibility of multibit reconstruction near criticality}
\label{sec:multibit}
In the previous section, we saw (Example~\ref{example:or-function}) that contrary to what may be the natural intuition, even restricting the messages to a single bit does not imply that a significant amount of information is destroyed at a particular level of reconstruction. In the 1-bit case, we overcame this obstacle using a multilevel analysis (of the iteration $(P^+_t,P^-_t) \mapsto (P^+_{t + 1},P^-_{t + 1})$) that treated the boundary of the 1-dimensional simplex specially. In the multibit case, the dynamics live in a higher-dimensional simplex and the boundary behavior appears very complicated to analyze (see Example~\ref{example:cycling-dynamics} in Appendix~\ref{apx:complicated}).

In this section, we give a new lower bound argument which completely overcomes this difficulty, proving the lower bound in Theorem~\ref{thm:main}. This argument requires several significant innovations, which we briefly summarize:

    \vspace{.1cm}
     \textbf{{Tracking only the law of the ``score'' $\E[X | Y]$:}} Let $X = X_{\rho}$ be the label of the root (of a depth $n$ tree) and $Y = Y_{\rho}$ be the reconstructed data at the root (i.e. the message that would be passed to an imaginary parent of the root node). The previous analysis tracked the complete distribution of $Y | X$. The recursion from the law for a depth $n$ tree to depth $n + 1$ tree is very easy to describe, but the resulting dynamics may be very complex in the multibit case (where the distributional recursion lives in a high-dimensional probability simplex). The new analysis considers only the induced law of $\E[X | Y]$ (i.e. the distribution of a natural real-valued random variable) and studies a BP-style recursion to relate the law at depths $n$ and $n + 1$. We remark that such an approach has been successfully used in many of the previous works around the reconstruction problem (see, e.g., \cite{BlRuZa:95,borgs2006kesten, pemantle2010critical})   
     The analysis of this nonlinear recursion is tamed by an approximate linearization and decoupling argument (Lemma~\ref{lem_approx}).
    
    \vspace{.1cm}
     \textbf{{Identifying attraction towards Gaussianity (in the natural Wasserstein metric):}} Recall that in the 1-bit case we were able to prove (Theorem~\ref{thm:inverse-theorem}) that ``good'' reconstruction functions (those that do not destroy much information) must push the law of $Y | X$ towards the middle of the 1-dimensional probability simplex. Analogously, we ultimately show (Lemma~\ref{lem_clt}) that good reconstruction functions push the law of $\E[X | Y]$ towards Gaussianity (measured in $W_2$ distance, see equation~\eqref{e_ng}).
     
    \vspace{.1cm}
    \textbf{{Multilevel analysis of kurtosis, variance, and Gaussianity:}} The proof of the key Gaussian attraction result (Lemma~\ref{lem_clt}) would be much simpler if, for $X_1,X_2$ i.i.d., the sum $\frac{1}{\sqrt{2}}(X_1 + X_2)$ were significantly closer to Gaussian (in $W_2$) than $X_1$ itself. Unfortunately, this is only true in the case of bounded kurtosis. Instead, we make a more complex argument with two main steps: (1) we first argue (Lemma~\ref{lem_fourth}) that the fourth moment is reasonably bounded after a sufficient number of reconstruction steps, and (2) give a multilevel tradeoff analysis showing that the kurtosis becomes large only when the variance of $\E[X | Y]$ shrinks significantly (i.e. information is destroyed).
Combining these ideas, we are able to show that any reconstruction algorithm on an alphabet of size $L$ which reconstructs all the way to the critical threshold would have to induce a distribution on $\E[X | Y]$ which is arbitrarily close to Gaussian --- but this is impossible for a distribution supported on $L$ atoms.\\\\
\textbf{Preliminaries: }Given an equiprobable $X\in \{\pm 1\}$ and an arbitrary random variable $Y$, denote 
the posterior mean by
\begin{align*}
 S_{X}(Y):=\E[X|Y]=\Pr(X = 1|Y)-\Pr(X = -1|Y).
\end{align*}
We remark that $S_{X}(Y)$ can be viewed as a discrete analogue of the score function in the estimation literature. Note that the probability of correctly reconstructing $X$ based on $Y$ equals $\frac{1}{2}\E\left[|S_X(Y)|\right]+\frac{1}{2}$.
The $\chi^2$-mutual information between $X$ and $Y$ equals 
$I_2(X;Y):=\E[S_X^2(Y)]$.
Since $S_X(Y)$ is bounded in $[-1,1]$, we have
$
\E[S_X^2(Y)]\le 
\E[|S_X(Y)|]\le \E^{1/2}[S_X^2(Y)]
$, so that as in \cite{BlRuZa:95, evans2000broadcasting},
the problem of solvability is reduced to bounding the $\chi^2$-mutual information.

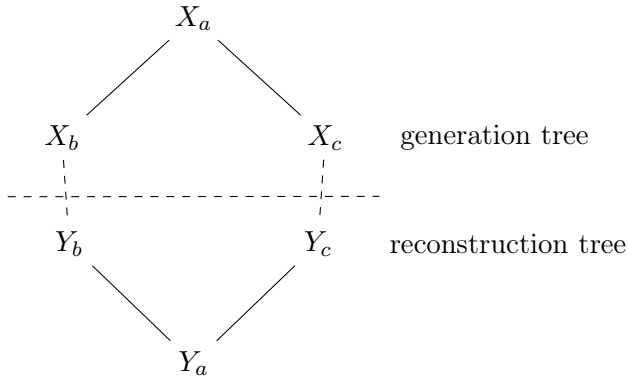
\begin{figure}[h!]
  \centering
\begin{tikzpicture}
[scale=1.75,
      dot/.style={draw,fill=black,circle,minimum size=0.7mm,inner sep=0pt},arw/.style={->,>=stealth}]
   \node[rectangle] (X) {$X_a$};
  \node[rectangle] (A) [below left= of X] {$X_b$};
  \node[rectangle] (B) [below right= of X] {$X_c$}; 
  \node[rectangle] (Y) [below=4cm of X] {$Y_a$};
  \node[rectangle] (U) [above left= of Y] {$Y_b$};
  \node[rectangle] (V) [above right= of Y] {$Y_c$};
    \node[rectangle] (C) [below left=0.5cm of A] {};
    \node[rectangle] (D) [below right=0.5cm of B] {};
    \node[rectangle] (E) [ right=0.5cm of B] {generation tree};
    \node[rectangle] (F) [ right=0.5cm of V] {reconstruction tree};
\draw (X) to (A);
\draw (X) to (B);
\draw [dashed](A) to (U);
\draw [dashed](B) to (V);
\draw (U) to (Y);
\draw (V) to (Y);
\draw [dashed](C) to (D);
\end{tikzpicture}
\caption{In this example, $S_{n-1}=S_{X_b}(Y_b)$, $S_n=S_{X_a}(Y_a)$, $\hat{S}_n=S_{X_a}(Y_b,Y_c)$, and $\bar{S}_n$ is an approximation of $\hat{S}_n$.}
\label{fig_1}
\end{figure}
Next, we define several key quantities in the proof of Theorem~\ref{thm:main}. 
For notational simplicity, 
we assume the tree is 2-ary ($d=2$) throughout the proofs,
noting that the result for general $d$ follows from the same argument.
Figure~\ref{fig_1} depicts parts of the generation tree and the mirroring reconstruction tree for node $a$ with children $b$ and $c$.
Note that $X_a$, $X_b$, $X_c$ denote binary labels and $Y_a$, $Y_b$, $Y_c$ denote reconstruction messages with values in $\Sigma$. 
Within any height $h$ tree (that is, from the root to a leaf there are $h$ edges) and for any $n\le h$, consider the following random variables:
\begin{itemize}
 \item $S_n$: defined as $S_{X_a}(Y_a)$ where $X_a$ is the height $n$ node in the generation tree and $Y_a$ is the mirroring height $n$ node in the reconstruction tree (recall \eqref{eqn:y_u}).
 \item $\hat{S}_n$: defined as $S_{X_a}(Y_b,Y_c)$ where $Y_b$ and $Y_c$ are the children of $Y_a$ in the reconstruction tree. 
 Note that $S_n$ is a conditional expectation of $\hat{S}_n$ (induced by applying $f_{a}$):  
\[
S_{X_a}(Y_a)=\E[X_a|Y_a] 
=\E[\E[X_a|Y_{a,b,c}]|Y_a]
=\E[\E[X|Y_{b,c}]|Y_a]
=\E[S_X(Y_b,Y_c)|Y_a].
\]
 \item $\bar{S}_n$: defined to be equal in distribution to $(1-2\varepsilon)(S_{n-1}+S'_{n-1})$,
where $S'_{n-1}$ is an independent copy of $S_{n-1}$. Since $S_{X_a}(Y_b) = (1-\varepsilon)S_{X_b}(Y_b) - \varepsilon S_{X_b}(Y_b) = (1-2\varepsilon)S_{X_b}(Y_b)$, one may view $\bar{S}_n$ as an idealized, decoupled version of $\hat{S}_n$. As will be shown in Lemma~\ref{lem_approx}, $\hat{S}_n$ and $\bar{S}_n$ are close in Wasserstein distance, which will allow us to carry out our analysis with the simpler quantity $\bar{S}_n$. 
\end{itemize}

Note that the subscript $n$ in all the random variables above denotes the height of the node in the generative tree (not the reconstruction tree).
Moreover, let $\sigma_n^2$, $\hat{\sigma}_n^2$, $\bar{\sigma}_n^2$, $\mu_n$, $\hat{\mu}_n$, $\bar{\mu}_n$ be the second and the fourth moments of these random variables.
The key of the proof is to study the evolution of the sequence $S_0\to \bar{S}_1\to \hat{S}_1\to S_1\to\dots$.

For any generative tree with height $h$ 
and any integer\footnote{We will also consider $L = \infty$, by which we mean there is no constraint on the alphabet size.} $L$, 
define 
\begin{align}\label{e_xi}
 \xi_h:=\sup\E[I_2(X;Y)],
\end{align}
where $X = X_{\rho}$ is the binary label on the root, $Y = Y_{\rho}$ is the final reconstruction, and the supremum is over all (possibly randomized) recursive reconstruction algorithms with memory $\log L$.
Here, we assume that the reconstruction functions at the same level are the same, but they are allowed to vary across levels. Similarly, for randomized algorithms, we assume that the distribution of the random reconstruction function at a node depends only on its level.

Since any randomized reconstruction algorithm for an $h+1$-level tree can be simulated by one for an $h$-level tree with the same memory, it follows
that $\xi_h$ monotonically decreases in $h$. 
Let $\xi:=\xi(\varepsilon,L)$ be the limit (which exists by monotonicity). 
Let $\varepsilon_c\in(0,1/2)$ be the supremum $\varepsilon$ for which
$\xi(\varepsilon,\infty)>0$.
As mentioned before, $d(1-2\varepsilon_c)^2=1$ is the KS bound. We now proceed to detail the steps in the proof of our main result:

\vspace{.1cm}
\textbf{Wasserstein approximation lemmas: }Recall that the idea of the converse proof is to show that if $\varepsilon$ is close to $\varepsilon_c$, then $S_n$ must converge to Gaussian for good algorithms.
This requires showing that $S_n$ is close to an i.i.d. sum.
While $\bar{S}_n$ is a bona fide i.i.d. sum (of $d$ independent copies of $S_{n-1}$, suitably scaled),
the distribution of $\hat{S}_n$ in relation to $S_{n-1}$ is more complicated.
However, it is possible to show that $\bar{S}_n$ and $\hat{S}_n$ are close near the critical threshold.
The idea is that the Wasserstein distance between $\bar{S}$ and $\hat{S}$ is small compared to their moments. A more general version of the following lemma is stated and proved in Section~\ref{sec:appendix-proof-linearization}.
\begin{lemma} 
\label{lem_approx}
With notation as above, we have that for any $n\in \{1,\dots,h-1\}$,
$$W_{2}^{2}(\hat{S}_{n+1}, \bar{S}_{n+1}) \leq \sigma_{n}^{2}\alpha_{2}(\sigma_{n}^{2}),$$
where $\alpha_{2}(\cdot)$ is a function satisfying $\lim_{x \to 0^{+}} \alpha_{2}(x) = 0$. A similar statement holds for $W^{4}_{4}$ with a suitable function $\alpha_{4}$, and with $\mu_{n}$ in place of $\sigma_{n}^{2}$. 
\end{lemma}

To use Lemma~\ref{lem_approx}, we need to upper-bound the moments of the score. For the second moment, we begin with the useful observation that in a tree of height $h$,
$\sigma_n^2\le \xi_n$ for any $n\in\{1,2,\dots,h\}$.  
The following result shows that, luckily,
there is a simple upper bound on $\xi(\varepsilon,L)$ which does not depend on $L$,
which in particular shows that the moments vanish when the noise is near the critical threshold.
\begin{prop}\label{prop_approx1}
Recall that we assumed that the degree $d=2$. 
For any $\varepsilon\in [0,1]$,
and $L\in\mathbb{N}\cup\{\infty\}$, either $\xi(\varepsilon,L) = 0$ or
$\xi(\varepsilon,L)\le \omega(\epsilon)$,
where we defined 
\begin{align}\label{e_omega}
\omega(\varepsilon)
:=4-\frac{2}{(1-2\varepsilon)^2}.
\end{align}
\end{prop}
In particular, Proposition~\ref{prop_approx1} implies the KS bound for reconstruction (including that the reconstruction problem is not solvable \emph{at} the threshold). The proof   follows from analysis of the standard BP recursion and is given in Section~\ref{subsec:approx1}; we note that a similar proof of the KS bound was previously discovered in \cite{borgs2006kesten}.

\vspace{.1cm}
\textbf{Bounding the fourth moment: }Wasserstein CLT results for i.i.d. sums are known under bounded fourth moment assumptions.
The following bound on the fourth moment of $S_n$ is derived from a recursive analysis given in Appendix~\ref{sec:appendix-fourth-moment-proof}.
\begin{lem}\label{lem_fourth}
There exists $\varepsilon_1\in(0,\varepsilon_c)$ such that for any $\varepsilon\in(\varepsilon_1,\varepsilon_c)$, 
there exist $h_1 = h_1(\varepsilon,L), h_2=h_2(\varepsilon,L)$ for which the following holds. For any
tree of height $h\ge h_2$ and any reconstruction algorithm, we have
either $\xi(\varepsilon,L) = 0$ or $\mu_n\le 13\xi^2(\varepsilon,L)$
at any level $n\in \{h_1,\dots, h\}$.
\end{lem}

\textbf{Normality for good algorithms near the threshold: }Given a real-valued random variable $Z$, let us define its Wasserstein non-Gaussianness by
\begin{align}
\mathcal{E}(Z):=\inf_{\sigma>0}W_2(Z,\,\E[Z]+\sigma G)
\label{e_ng}
\end{align}
where $G$ is the standard Gaussian random variable. The following lemma, exploiting the Wasserstein CLT, is proved in Section~\ref{sec:appendix-proof-clt}. 
\begin{lem}\label{lem_clt}
There exists $\varepsilon_2\in(0,\varepsilon_c)$ such that for any $\varepsilon\in(\varepsilon_2,\varepsilon_c)$,
$L\in\{1,2,\dots\}\cup\{\infty\}$,
and $\delta\in(0,1/2)$, 
either $\xi(\varepsilon,L) = 0$ or
\begin{align*}
\lim_{h\to\infty}\sup_{\textrm{algorithms}\colon \sigma_h^2\ge(1-\delta)\xi(\varepsilon,L)}\mathcal{E}(S_h)
\le c_4\sqrt{\xi(\varepsilon,L)}\left(
(\varepsilon_c-\varepsilon)^{1/13}
+\sqrt{\delta\log\frac{1}{\varepsilon_c-\varepsilon}}
\right),
\end{align*}
where $c_4$ is an absolute constant.
\end{lem}
We are finally in the position of proving the lower bound:

\begin{proof}[Proof of the lower bound in Theorem~\ref{thm:main}]
Consider any $\varepsilon\in(\varepsilon_2,\varepsilon_c)$ and $L\in\{1,2,\dots\}$, where $\varepsilon_2$ is as defined in Lemma~\ref{lem_clt}.
Choose any $\delta\in(0,1/2)$.
Note that for any $h$, there exists an algorithm such that $\sigma_h^2\ge (1-\delta)\xi$.
Lemma~\ref{lem_quant} (proved in Appendix~\ref{sec:appendix-proof-non-gaussian-bound}) shows that $\mathcal{E}(S_h)\ge \frac{1}{2L}$.
Comparing with the result in Lemma~\ref{lem_quant}
we have
$c_4\left(
(\varepsilon_c-\varepsilon)^{1/13}
+\sqrt{\delta\log\frac{1}{\varepsilon_c-\varepsilon}}
\right)\ge \sqrt{1-\delta}/2L$.
Taking $\delta\downarrow 0$ establishes that
there exists an absolute constant $c$ such that if $\xi(\varepsilon,L)>0$ then $L\ge c(\varepsilon_c-\varepsilon)^{-1/13}$.
\end{proof}
We conclude this section by noting that Corollary~\ref{corr:bp-fixpoint} is proved in Appendix~\ref{sec:bp-fixpoint}, by combining the above with a simple argument to prove distributional convergence of BP.

\newpage
\bibliographystyle{plain}
\bibliography{all,reconstruction}

\newpage
\appendix

\section{Impossibility of 1-Bit Reconstruction near criticality}
\label{sec:1-bit}
This appendix provides complete proofs for various statements sketched in Section~\ref{sec:1-bit-sketch}; the organization mirrors that of Section~\ref{sec:1-bit-sketch}.
\subsection{A Direct Proof of the Kesten-Stigum Bound}
\label{sec:appendix-direct-KS-bound}
As mentioned in the introduction, the KS bound for general trees is stated in terms of the branching number, which was introduced in \cite{Lyons:90}. This notion admits several equivalent definitions; for us, the most convenient is the following.  
\begin{definition} 
\label{defn:branching-number}
The \emph{branching number} of a rooted tree $T$ with root $\rho$ is defined as
$$\text{br}(T) := \sup_{\lambda \geq 1} \left\{\inf_{\Pi}\sum_{v\in \Pi} \lambda^{-|v|} > 0\right\},$$
where the infimum is over all cutsets $\Pi$ (a set
of vertices of $T\setminus\{\rho\}$ is called a cutset if it intersects every infinite path emanating
from $\rho$), and $|v|$ denotes the number of edges on the path from $v$ to $\rho$.  
\end{definition}


The following key `mixing inequality' was discussed in Section~\ref{subsec:KS-direct-proof-sketch}.
\begin{lemma}\label{lem:mixing-inequality}
Suppose $P$ and $Q$ are distributions on the same alphabet $\Sigma$. For $\delta \in [0,1/2]$, define the mixture distributions
$$P(\delta) := \left(\frac{1}{2} + \delta\right)P + \left(\frac{1}{2} - \delta\right)Q$$ and 
$$Q(\delta) := \left(\frac{1}{2} - \delta\right)P + \left(\frac{1}{2} + \delta\right)Q.$$
Then, for any $\delta' \geq  \delta$ with $\delta' \in (0,1/2]$, we have
\[ \SKL(P(\delta),Q(\delta)) \le (\delta/\delta')^2 \SKL(P(\delta'),Q(\delta')). \]
In particular, 
$$\SKL(P(\delta),Q(\delta)) \le 4 \delta^2 \SKL(P,Q).$$
\end{lemma}
\begin{proof}
We may assume that $\SKL(P(\delta'),Q(\delta')) < \infty$, as there is nothing to prove otherwise.
Moreover, by a standard approximation argument, it suffices to prove the claim for discrete alphabets $\Sigma$. In this case, note that 
$$\SKL(P,Q)=\sum_{a\in \Sigma}(p_a - q_a)(\log p_a - \log q_a).$$
Consider the change of variables 
$$p_a = s_a + t_a$$ and 
$$q_a = s_a - t_a,$$ under which 
$$p(\delta)_a = s_a + 2\delta t_a $$ and $$q(\delta)_a = s_a - 2\delta t_a.$$ Note that we must necessarily have $2\delta'|t_a| < s_a$; otherwise, there would exist some $a \in \Sigma$ for which $p(\delta')_a\neq 0$ but $q(\delta')_a = 0$ or vice versa, thereby contradicting the assumed finiteness of $\SKL(P(\delta'),Q(\delta'))$. Therefore, we see that  
\begin{align*}
\SKL(P(\delta),Q(\delta)) 
&= \sum_{a \in \Sigma} 4 \delta t_a (\log (s_a + 2\delta t_a) - \log (s_a - 2\delta t_a)) \\
&= 4\sum_{a \in \Sigma} \delta t_a (\log (1 + 2\delta t_a/s_a) - \log (1 - 2\delta t_a/s_a))\\
&= 8 \sum_{a \in \Sigma} s_a \left(2(\delta t_a/s_a)^{2} + \frac{8}{3}(\delta t_a/s_a)^{4} + \cdots\right),
\end{align*}
where the last equality follows from the power series expansion of $\log(1+x)$ (valid for $|x|<1$) around $x=0$. 
Finally, the result follows from the term-wise observation that for any $k \ge 1$ and $\delta \le \delta'$,
$$\delta^{2k} = \left(\frac{\delta}{\delta'}\right)^{2k} (\delta')^{2k} \le \left(\frac{\delta}{\delta'}\right)^{2} (\delta')^{2k}.$$ 
\end{proof}

\begin{remark}
Observe for comparison that the joint convexity of SKL only gives the bound 
\begin{align*}
\SKL(P(\delta),Q(\delta)) &= 
\SKL\left(\left(1-2\delta\right)\left(\frac{P}{2} + \frac{Q}{2}\right) + 2\delta P, \left(1-2\delta\right)\left(\frac{P}{2} + \frac{Q}{2}\right) + 2\delta Q\right)\\
&\le 2\delta \SKL(P,Q), 
\end{align*}
which is much weaker for $\delta < 1/2$.
\end{remark}
\begin{remark}
The proof shows that the above inequality is asymptotically tight as $\delta,\delta' \to 0$ for any fixed distributions $P$ and $Q$. 
\end{remark}
\begin{remark}
A similar proof can also be used to establish the same inequality for $H^{2}$, the squared Hellinger distance.  
\end{remark}

We are now ready to present a proof of the general KS bound. Since the proof is essentially the same as for the infinite $d$-ary tree, we will only sketch the details. 
\begin{theorem}[Kesten-Stigum bound]
\label{thm:KS-general-trees}
If $4\text{br}(T)\nu^{2} < 1$, then $\TV(T_{n}^{+},T_{n}^{-})\to \infty$ as $n\to \infty$. 
\end{theorem}
\begin{proof}
We will find it more convenient to use the Hellinger-squared distance $\Hel^{2}$ instead of $\SKL$ owing to the fact that $\Hel^{2}(P,Q)\leq 1$ for all distributions $P$ and $Q$. Since $\TV(P,Q) \leq \sqrt{2}\Hel(P,Q)$, it clearly suffices to show that if $4\text{br}(T)\nu^{2} < 1$, then $\inf_{\Pi}\Hel(P_\Pi^+, P_\Pi^-) = 0$. For this, let $\Pi$ be a cutset, let $\rho_1,\dots,\rho_{d_1}$ denote the children of $\rho$, and let $\Pi_{1},\dots,\Pi_{d_1}$ denote the intersections of $\Pi$ with the descendants of $\rho_{1},\dots,\rho_{d_1}$. 
Then, since $\Hel^{2}$ satisfies the same `mixing inequality' as $\SKL$, and since
$\Hel^{2}(P^{\otimes d},Q^{\otimes d}) \leq d\Hel^{2}(P,Q),$
the same argument as in the $d$-ary case shows that
$$\Hel(P^{+}_{\Pi}, P^{-}_{\Pi})^{2} \leq 4\nu^{2}\sum_{i=1}^{d_1}\Hel(P^{+}_{\Pi_i}, P^{-}_{\Pi_i})^{2},$$
where we think of $\Pi_i$ as a cutset of the subtree rooted at $\rho_i$. Iterating this process until all the roots under consideration lie in $\Pi$ (this is guaranteed to happen by the definition of a cutset), we find that
$$\Hel(P^{+}_{\Pi}, P^{-}_{\Pi})^{2} \leq \sum_{v \in \Pi}(4\nu^{2})^{|v|}.$$
Finally, taking the infimum over both sides, and using the definition of the branching number, completes the proof.
\end{proof}

\subsection{Interlude: noisy-message passing algorithms fail near criticality}
For this subsection and the next, it will be convenient to formally establish some notation which has already been discussed in Section~\ref{sec:1-bit-sketch}. We restrict ourselves to $d$-ary trees, and label the levels of the $n$-level $d$-ary tree in decreasing order, with the root being level $n$ and the leaves being level $0$. We also restrict ourselves to message-passing algorithms for which the reconstruction function depends only on the level of the tree i.e. $f_u = f_\ell$ for node $u$ at level $\ell$ of the tree. As in Section~\ref{subsec:noisy-message-sketch}, let $P_n^{\pm}$ denote the distribution
on $\Sigma$ (which we interpret as the final message received at
the root) obtained by broadcasting on an $n$ level $d$-ary tree, conditioned on
the root being $\pm$, and then reconstructing using our message-passing
algorithm. 

Then, in the case when each message passes through an independent copy of a noisy channel $P_{Y|X}: \Sigma \to \Sigma$, we have from the description of the broadcast and reconstruction processes that 
\[
P_{n}^{\pm}=(f_{n})_{*}\left(\left(P_{Y|X}\circ\left(\left(\frac{1}{2}+\nu\right)P_{n-1}^{\pm}+\left(\frac{1}{2}-\nu\right)P_{n-1}^{\mp}\right)\right)^{\otimes d}\right),
\]
where $f_{*}(\mu)$ denotes the pushforward of the measure
$\mu$ by the function $f$, and $P_0^{\pm}$ are initial states specified by the message passing scheme.

We can now state and prove our general theorem on the impossibility of reconstruction near criticality by noisy message-passing algorithms, as discussed in Section~\ref{subsec:noisy-message-sketch}. 
\begin{theorem}
\label{thm:KS-noisy-general}
Suppose that $P_{Y | X}$ satisfies an SDPI with constant $\eta < 1$. If $4d \nu^2 \eta < 1$, then reconstruction under the multilevel noise model is impossible for any message passing algorithm. 
\end{theorem}
\begin{proof}
For distributions $P_n^{\pm}$ and $\nu \in (0,1/2)$, we will use the notation $P_{n}^{\pm}(\nu)$ from Lemma~\ref{lem:mixing-inequality}. Then, similar to the proof of the KS bound for $d$-ary trees, we have
\begin{align*}
    \SKL(P_n^{+}, P_n^{-})
    &\le \SKL\left((P_{Y|X}\circ P_{n-1}^{+}(\nu))^{\otimes d}, (P_{Y|X}\circ P_{n-1}^{-}(\nu))^{\otimes d}\right) \\ 
    & = d \SKL(P_{Y | X} \circ P_{n - 1}^{+}(\nu),  P_{Y | X} \circ P_{n - 1}^{-}(\nu)) \\
    &\le \eta d \SKL(P_{n - 1}^{+}(\nu), P_{n - 1}^{-}(\nu))\\
    &\le \eta 4d\nu^2 \SKL(P_{n - 1}^{+}, P_{n - 1}^{-}).
\end{align*}
Iterating this inequality and using $\SKL(P_1^{+},P_1^{-}) < \infty$, we see that $\SKL(P_n^{+},P_n^{-}) \to 0$ as $n \to \infty$.
\end{proof}
We conclude with a couple of examples of common channels which satisfy an SDPI. 
\begin{example}
\label{example:mixture-noise}
For a fixed (noise) distribution $\mu$, the channel $P_{Y|X}$ given by $P_{Y|X}\circ P = (1-\delta)P+\delta \mu$ obeys an SDPI with $\eta \leq (1-\delta)$, as is seen by the joint convexity of SKL:
$\SKL((1 - \delta) P + \delta \mu, (1 - \delta) Q + \delta \mu) \le (1 - \delta) \SKL(P,Q)$.
\end{example}

\begin{example}
\label{example:bp-additive-noise}
For a real valued random variable $X$, let $X' = X+\delta Z$, where $Z \sim N(0,1)$ is independent of $X$, and let $Y=g(X')$, where $g(x) = -1$ if $x\leq -1$, $g(x)=1$ if $x\geq 1$, and $g(x)=x$ otherwise. The channel $P_{Y|X}$, which corresponds to adding a small Gaussian noise and then thresholding the messages to lie in $[-1,1]$ (as required for belief propagation) also satisfies an SDPI: for any distributions $P,Q$ on $[-1,1]$,
$\SKL(P_{Y|X}\circ P, P_{Y|X}\circ Q) \leq \SKL(P_{X'|X}\circ P, P_{X'|X}\circ Q) \leq (1-2F(1/\delta))\SKL(P,Q)$, where $F(x) = 1-\Phi(x)$ is the standard Gaussian complementary CDF. Here, the first inequality is the usual DPI, and the second inequality follows from \cite{polyanskiy2016dissipation}.
\end{example}
\subsection{1-bit message-passing algorithms fail near criticality}
The initial part of our discussion in this subsection holds for any finite alphabet $\Sigma$. Later on, we will specialise our discussion to the 1-bit setting i.e. when $|\Sigma|=2$. \\
\textbf{Restricted SDPI for discrete functions: }As discussed in Section~\ref{subsec:1-bit-sketch}, while general discrete functions $f\colon \Sigma^{d} \to \Sigma$ need not satisfy an SDPI, we can obtain such an inequality provided we restrict the class of input distributions to those which are `robustly' of full support. More precisely, we have the following.
\begin{theorem}[Restricted SDPI for Discrete Functions]\label{thm:restricted-sdpi-multibit}
Fix $d$ and $\gamma > 0$. 
Suppose that $|\Sigma| \le d$.
Let $\Delta_{\gamma}$ be the subset of the probability
simplex in $\mathbb{R}^{\Sigma}$ given by requiring for $p \in \Delta$ that $p_{a} \ge \gamma$ for all $a \in \Sigma$. 
Then there exists $\eta = \eta(|\Sigma|,\gamma,d) < 1$ such that 
\[ \max_{f : \Sigma^d \to \Sigma} \sup_{P,Q \in \Delta_{\gamma}} \frac{\KL(f_*(P^{\otimes d}), f_*(Q^{\otimes d}))}{\KL(P^{\otimes d},Q^{\otimes d})} \le \eta < 1. \]
\end{theorem}
\begin{proof}
As there are only finitely many functions $f:\Sigma^{d}\to \Sigma$, it suffices to show the desired inequality
for a fixed (but otherwise arbitrary) function $f$. Note that any two distributions $P,Q \in \Delta_{\gamma}$ have full support in $\Sigma$, and hence, $\KL(P,Q)<\infty$ (in fact, $\KL(P,Q) \leq \log({1/\gamma})$).  
Moreover, for any $P \ne Q$ such that $\KL(P,Q)<\infty$, it follows from the strict case of the data processing inequality (using the assumption that $|\Sigma| \leq d$) that $$\KL(f_*(P^{\otimes d}), f_*(Q^{\otimes d})) < \KL(P^{\otimes d}, Q^{\otimes d}) = d \KL(P,Q).$$ This suggests using the compactness of $\Delta_{\gamma}$ to obtain the desired inequality; in order to be able to do this, it only remains to show that for any $P\in \Delta_{\gamma}$, $$\limsup_{Q \to P}g(P,Q) < 1,$$ where 
$$g(P,Q) :=  \frac{\KL(f_*(P^{\otimes d}), f_*(Q^{\otimes d}))}{d \KL(P,Q)}.$$  

Accordingly, fix $P\in \Delta_{\gamma}$, and define $Q$ by $q_a = p_a - \delta_a$, where $|\delta|_{\infty} \le \gamma/2$ and $\sum_{a\in \Sigma}\delta_{a} = 0$. 
Then, we have 
\begin{align}
\KL(P,Q) 
&= \sum_{a \in \Sigma} p_a \log \frac{p_a}{q_a} \nonumber \\
&= \sum_{a \in \Sigma} p_a \log\left(1 + \frac{\delta_a}{q_a}\right) \nonumber\\
&= \sum_{a \in \Sigma} (q_a + \delta_a) \left(\frac{\delta_a}{q_a} - \frac{\delta_a^2}{2q_a^2} + O_{\gamma}(|\delta|_{\infty}^3)\right) \nonumber \\
&= \sum_{a \in \Sigma} \left(\delta_a + \frac{\delta^2_a}{2 q_a} +  O_{\gamma}(|\delta|_{\infty}^3)\right) \nonumber \\ 
&= \sum_{a \in \Sigma} \frac{\delta_a^2}{2 q_a} + O_{\gamma,d}(|\delta|_{\infty}^3). \label{eqn:kl-quadratic}
\end{align}
Since $P^{\otimes d}$ and $Q^{\otimes d}$ have full support in $\Sigma^{d}$, it follows that $P':=f_*(P^{\otimes d})$ and $Q':=f_*(Q^{\otimes d})$ share the same support $\Sigma' \subset \Sigma$. 
Moreover, since the image of a compact set under a continuous map is compact, $K := f_*(\{P^{\otimes d} : P \in \Delta_{\gamma}\})$ is also a compact set; since each point in $K$ is separated by a positive distance from the boundaries of the simplex $\Delta'$ of probability distributions supported on $\Sigma'$, it follows by compactness that $K\subset \Delta'_{\gamma'}$ for some $\gamma' > 0$. Now, writing $p'_a = q'_a + \delta'_a$ with $|\delta'|_{\infty}$ sufficiently small, the same calculation as above shows that
\begin{align}
\KL(P',Q') 
&= \sum_{a \in \Sigma'} \frac{(\delta'_a)^2}{2 q'_a} + O_{\gamma,d}(|\delta'|_{\infty}^3) \nonumber \\
&= \sum_{a \in \Sigma'} \frac{(\delta'_a)^2}{2 q'_a} + O_{\gamma,d}(|\delta|_{\infty}^3),
\label{eqn:kl-quadratic-image}
\end{align}
where the last equality uses 
$$|\delta'|_{\infty} \leq \TV(P',Q') \leq \TV(P^{\otimes d},Q^{\otimes d}) \leq d\TV(P,Q) \leq d^{2}|\delta|_{\infty}.$$ 
We now analyze the leading term in \eqref{eqn:kl-quadratic-image}. We will need the following preliminary computation:  
for any function $h : \Sigma^d \to [-1,1]$,
\begin{align*}
\E_{P^{\otimes d}}[h] - \E_{Q^{\otimes d}}[h] 
&= \sum_{x\in \Sigma^{d}} h(x) ((Q + \delta)^{\otimes d}(x) - Q^{\otimes d}(x)) \\
&= \sum_{i=1}^{d} \sum_{x\in \Sigma^{d}} h(x) Q^{\otimes d-1}(x_{\sim i}) \delta_{x_i} + O_{d}(|\delta|_{\infty}^2)\\
&= \sum_{i=1}^{d} \E_{X\sim Q^{\otimes d}}\left[h(X) \frac{\delta_{X_i}}{q_{X_i}}\right] + O_{d}(|\delta|_{\infty}^2).
\end{align*}
Therefore, if $V$ is the subspace of $L^2(Q)$ spanned by $(1_{f = \ell})_{\ell \in \Sigma'}$, we find that for $X \sim Q^{\otimes d}$
\begin{align*}
&\sum_{a\in \Sigma'}\frac{(\delta_a')^{2}}{2q_a'} = \sum_{l\in \Sigma'}\frac{(\E_{P^{\otimes d}}[1_{f = \ell}] - \E_{Q^{\otimes d}}[1_{f = \ell}])^2}{\sqrt{2}\E_{Q^{\otimes d}}[1_{f = \ell}]} \\
&= \sum_{l\in \Sigma'} \left(\E_X
\left[ \frac{1_{f = \ell}}{\sqrt{\Pr_{Q^{\otimes d}}[f = \ell]}} \sum_{i=1}^{d} \frac{\delta_{X_i}}{\sqrt{2}q_{X_i}}\right]\right)^2 + O_{\gamma,d}(|\delta|_{\infty}^3) \\
&= \E_{X}\left[\left(\sum_{i=1}^{d} \frac{\delta_{X_i}}{\sqrt{2}q_{X_i}}\right)^2\right]
- \left\|\Proj_{V^{\perp}} \sum_{i=1}^{d}\frac{\delta_{X_i}}{\sqrt{2}q_{X_i}}\right\|_{L^2(Q)}^2
+ O_{\gamma,d}(|\delta|_{\infty}^3)  \\
&= \E_{X}\left[ \sum_{i=1}^{d} \frac{\delta_{X_i}^2}{2q_{X_i}^2}\right]
 - \left\|\Proj_{V^{\perp}} \sum_{i=1}^{d}\frac{\delta_{X_i}}{\sqrt{2}q_{X_i}}\right\|_{L^2(Q)}^2
 + O_{\gamma,d}(|\delta|_{\infty}^3)&,
\end{align*}
where the second equality is by the Pythagorean theorem in $L^2(Q)$, and the last is by expanding the square and using $\sum_{a\in \Sigma} \delta_a = 0$. 

Since the random variable $\sum_{i=1}^{d}\delta_{X_i}/(\sqrt{2}q_{X_i})$ takes on at least $d + 1$ distinct values, and since any random variable in the span of $(1_{f = \ell})_{\ell \in \Sigma'}$ can take on at most $|\Sigma'| \le d$ distinct values, it follows that $$\left\|\Proj_{V^{\perp}} \sum_{i=1}^{d}\frac{\delta_{X_i}}{\sqrt{2}q_{X_i}}\right\|_{L^2(Q)} > 0,$$
and hence,
$$\frac{\E_{X \sim Q^{\otimes d}}\left[ \sum_{i=1}^{d} \frac{\delta_{X_i}^2}{2q_{X_i}^2}\right]
 - \left\|\Proj_{V^{\perp}} \sum_{i=1}^{d}\frac{\delta_{X_i}}{\sqrt{2}q_{X_i}}\right\|_{L^2(Q)}^2}{\E_{X\sim Q^{\otimes d}}\left[ \sum_{i=1}^{d} \frac{\delta_{X_i}^2}{2q_{X_i}^2}\right]} < 1.$$
Note that the above quantity depends continuously on the (Euclidean) unit vector in the direction of $(\delta_{a})_{a\in \Sigma}$, viewed as a vector in the hyperplane of $\R^{\Sigma}$ given by $\{(x_a)_{a\in \Sigma}\in \R^{\Sigma}: \sum_{a}x_a = 0\}$. Hence, it follows by the compactness of the unit sphere in finite dimensions that
\[ \frac{\E_{X \sim Q^{\otimes d}}\left[ \sum_{i=1}^{d} \frac{\delta_{X_i}^2}{2q_{X_i}^2}\right]
 - \left\|\Proj_{V^{\perp}} \sum_{i=1}^{d}\frac{\delta_{X_i}}{\sqrt{2}q_{X_i}}\right\|_{L^2(Q)}^2}{\E_{X\sim Q^{\otimes d}}\left[ \sum_{i=1}^{d} \frac{\delta_{X_i}^2}{2q_{X_i}^2}\right]} < c,\]
 for $c < 1$ independent of $\delta$. Finally, since 
 $$\E_{X \sim Q^{\otimes d}}\left[ \sum_{i=1}^{d} \frac{\delta_{X_i}^2}{2q_{X_i}^2}\right] = d\cdot \sum_{a\in \Sigma}\frac{\delta_a^{2}}{2q_a} = d\KL(P,Q) + O_{\gamma,d}(|\delta|_{\infty}^{3})$$
by \eqref{eqn:kl-quadratic}, the desired conclusion follows by taking the limit as $|\delta|_{\infty}\to 0$. 
\end{proof}

\textbf{Inverse SKL non-contraction theorem for Boolean functions: }For the remainder of this section, we will restrict to the case when $|\Sigma|=2$. In this case, we are able to go beyond the restricted SDPI, and show that the only situation where we do not have contraction in SKL is essentially that of Example~\ref{example:or-function}. We begin with a preliminary definition.

\begin{definition}
\label{defn:OR-like}
A Boolean function $f : \{0,1\}^d \to \{0,1\}$ is \emph{OR-like} if $f(0,\ldots,0) \ne f(1,0,\ldots, 0)$
and $f(1,0,\ldots,0) = f(0,\ldots,0,1,0,\ldots,0)$ i.e. $f$ is constant on inputs with hamming weight $1$.
A Boolean function $f$ is \emph{AND-like} if $g$ is OR-like where $g(x) = 1 - f(1 - x_1, \ldots, 1 - x_d)$.
\end{definition}
\begin{theorem}[Inverse Theorem for SKL non-contraction]\label{thm:inverse-theorem}
Fix $d > 1$ and $C > 0$. There exist constants $\Delta = \Delta(d,C) > 0$ and $\eta  = \eta(d,C) > 0$ such that for any Boolean function $f : \{0,1\}^d \to \{0,1\}$, and $P=\Ber(p),Q = \Ber(q)$ with $\SKL(P,Q) < C$, if
\[ \frac{\SKL(f(P^{\otimes d}), f(Q^{\otimes d}))}{d \SKL(P,Q)} > 1 - \Delta, \]
then $f$ must be OR-like with $p,q \leq \eta$ or AND-like with $p,q \ge 1 - \eta$.
\end{theorem}
\begin{proof}
First, by Theorem~\ref{thm:restricted-sdpi-multibit}, we know that
\[ \sup_{p,q \in [\delta,1 - \delta]}\frac{\KL(f(P^{\otimes d}), f(Q^{\otimes d}))}{d \KL(P,Q)} < 1 \]
for all $\delta > 0$,
and by symmetrizing we get the same statement for SKL. Therefore, it remains to analyze what happens when $p$ and $q$ approach $0$ or $1$ together (note that it cannot be the case that $p$ approaches $0$ and $q$ approaches $1$, or vice versa, by our assumption that SKL is bounded by $C$). 

By symmetry, it suffices to consider the case $p \to 0, q \to 0$. After possibly replacing $f$ by $1-f$, we may further assume that $f(0,\ldots,0) = 0$. Let $k$ be the number of inputs $x$ of Hamming weight 1 such that $f(x) = 1$; if $f$ is not OR-like, then $k < d$. Let 
$$p' = \Pr_{P^{\otimes d}}(f = 1)$$ and 
$$q' = \Pr_{Q^{\otimes d}}(f = 1).$$ Then,
$$p' - q' = k(p - q) + o(p - q),$$ 
and using 
$$|\log(1 - p) - \log(1 - q)| \le 2 |p - q|$$ for 
$p,q$ sufficiently small by the mean value theorem, we see that
\[ \SKL(\Ber(p),\Ber(q)) = (p - q)(\log(p) - \log(q)) + O((p - q)^2) \]
and similarly for $\SKL(\Ber(p'),\Ber(q'))$.
Therefore,
\[ \lim_{p,q\to 0} \frac{\SKL(\Ber(p'),\Ber(q'))}{d\SKL(\Ber(p),\Ber(q)} = \frac{k}{d}, \]
which is less than $1$ if $f$ is not OR-like.
Hence, by compactness of the probability simplex, there exists a $\Delta$ such that 
$$\frac{\SKL(f(P^{\otimes d}), f(Q^{\otimes d}))}{d\SKL(P,Q)} < 1 - \Delta,$$
which completes the proof.
\end{proof}

We note that the assumption $\SKL(P,Q) < C$ in the above theorem is trivially satisfied for our applications, due to the following simple observation. Here, $P_{n}^{\pm}$ refer to the distributions defined earlier. 
\begin{lemma}\label{lem:skl-bounded}
For any $\nu \le 1/4$ and $n \ge 1$,
\[ \SKL(P_n^{+}, P_n^{-}) \le d \SKL(\Ber(3/4),\Ber(1/4)) \]
\end{lemma}
\begin{proof}
Fix $n$ and consider the broadcasting and reconstruction process on the $d$-regular tree of depth $n$. Let $X_{\rho}$ be the label of the root, $X_{N(\rho)}$ be the labels of the direct children of the root, and $Y$ the output of the reconstruction process (so $P_n^{\pm}$ is the law of $Y$ given $X_{\rho} = \pm$). Since $Y$ is conditionally independent of $X_{\rho}$ given $X_{N(\rho)}$, we have by the  data processing inequality that
\begin{align*}
\SKL(P_n^{+},P_n^{-}) 
&\le \SKL\left(\L(X_{N(\rho)} | X_{\rho} = +), \L(X_{N(\rho)} | X_{\rho} = -)\right)\\
&= d \SKL(\Ber(1/2 + \nu), \Ber(1/2 - \nu)).
\end{align*}
Since $\nu \le 1/4$ by assumption, the desired bound follows.
\end{proof}

\textbf{Failure of 1-bit message-passing algorithms with globally fixed reconstruction function: }The previous considerations, together with a simple case analysis, allow us to show that 1-bit message-passing algorithms which use the \emph{same} reconstruction function at every node fail to solve the reconstruction problem near criticality, thereby extending the main result in \cite{mossel1998recursive}.
\begin{theorem}\label{thm:no-reconstruction-fixed-function}
There exists $\nu_1$ such that $4d \nu_1^2 > 1$ (i.e. reconstruction is information-theoretically possible) but no 1-bit reconstruction algorithm with fixed reconstruction function $f$ solves the reconstruction problem for $\nu \le \nu_1$. 
\end{theorem}
\begin{proof}
For the analysis, define $g(\rho) = \E_{\Ber(\rho)^{\otimes n}}[f]$. We know reconstruction is impossible if $f$ is constant, so w.l.o.g. we may assume that $f$ is not constant. Recall that $\SKL(p_t,q_t)$ is upper bounded by a constant due to Lemma~\ref{lem:skl-bounded}.

The proof proceeds by case analysis on the boundary behavior of $f$. We give the analysis of the first case in explicit detail and then describe the modifications to this analysis for each of the other cases. Here $\vec 0 = (0,\ldots,0)$ and likewise for $\vec 1$. 
\begin{enumerate}
    \item $f(\vec 0) = 0$, $f(\vec 1) = 1$. If $f$ is OR-like then $g(\rho) > \rho$ for all $\rho < \rho_0$, so there is some neighborhood of $0$ ($\rho < \rho_0^n$) which the dynamics will not enter\footnote{More precisely, no iterate $p_n$ or $q_n$ for $n \ge 1$ will lie in this neighborhood.}. If $f$ is AND-like, then the same holds in a neighborhood around $1$. Applying Theorem~\ref{thm:inverse-theorem}, we see the SKL contracts by at least some absolute constant in each step in all the regions the dynamics can enter. 
    \item $f(\vec 0) = f(\vec 1)$. By symmetry assume $f(\vec 0) = 0$. Then, there is a neighborhood of $1$ which the dynamics will not enter. Additionally, if $f$ is OR-like, the same is true of $0$. In any case, we see from Theorem~\ref{thm:inverse-theorem} that the SKL contracts by at least some absolute constant in each step in all the regions the dynamics can enter. 
    \item $f(\vec 0) = 1$, $f(\vec 1) = 0$. If $f$ is not OR-like or AND-like, then we are done by applying Theorem~\ref{thm:inverse-theorem}. If $f$ is OR-like around $0$, then there is a neighborhood of $1$ which will not be entered, hence there is also a neighborhood of $0$ which will also not be entered; once again, apply Theorem~\ref{thm:inverse-theorem} to see that the SKL contracts by at least some absolute constant in each step in all the regions the dynamics can enter on the complement. The case when $f$ is AND-like around $1$ is handled similarly.
\end{enumerate}
In every case, we showed that the dynamics stay within a region where we gain some absolute constant factor in the data processing inequality (by Theorem~\ref{thm:inverse-theorem}). Hence, by the same argument as in Theorem~\ref{thm:KS-noisy-general} we see that reconstruction is impossible sufficiently close to the KS threshold. 
\end{proof}

\textbf{Failure of general 1-bit message-passing algorithms: }As discussed in Section~\ref{subsec:1-bit-sketch}, the analysis of Theorem~\ref{thm:no-reconstruction-fixed-function} relies strongly on the fact that the function $f$ is fixed throughout the tree, and does not extend to exclude natural 1-bit message-passing algorithms where reconstruction functions vary across levels. The next two theorems make the sketch from Section~\ref{subsec:1-bit-sketch} precise. 


\begin{theorem}
\label{thm:regularized-sdpi}
Fix $d>1$ and $C>0$. Let $P = \Ber(p), Q = \Ber(q)$ such that $\SKL(P,Q) \le C$ and $p,q \in (0,1)$. 
There exists $\lambda = \lambda(d,C) > 0$ independent of $P$ and $Q$ such that, defining
\[ \phi(P,Q) := \log \SKL(P,Q) - \lambda \left(\log \frac{p + q}{2} + \log \left(1 - \frac{p + q}{2}\right)\right), \]
then there exists $c = c(d,C) > 0 $ such that
\[ \phi(f_*(P^{\otimes d}), f_*(Q^{\otimes d})) - \log d \le L(P,Q) - c \]
for all non-constant functions $f:\{0,1\}^{d}\to \{0,1\}$. 
\end{theorem}
\begin{proof}
As before, by the assumed upper bound of the SKL between the distributions under consideration, we only need to consider the cases when $p,q\to 0$, $p,q\to 1$, and when both $p,q$ are bounded away from $0$ and $1$. The analysis of the first two cases is identical, so we will only consider the case when $p,q\to 0$. As in the proof of Theorem~\ref{thm:inverse-theorem}, let $k$ be the number of inputs $x$ of Hamming weight $1$ for which $f(x) \ne f(\vec 0)$. W.l.o.g. we may assume $f(\vec 0) = 0$. Assuming that $k\geq 1$, by the same calculation as in the proof of Theorem~\ref{thm:inverse-theorem}, we have
\begin{align*} 
\phi(f_*(P^{\otimes d}), f_*(Q^{\otimes d}))
&= \log \left(k \SKL(P,Q)
- \lambda \left(\log \frac{k(p + q)}{2} + \log \left(1 - \frac{k(p + q)}{2}\right)\right)\right) + o(p+q) \\
&= (1 - \lambda) \log{k} + \log \SKL(P,Q)
- \lambda \left(\log \frac{p + q}{2} + \log \left(1 - \frac{p + q}{2}\right)\right) + o(p+q),
\end{align*}
so that after subtracting $\log d$, we find a decrease of $(1 - \lambda) \log{k} - \log{d}$, which translates to a decrease by an absolute constant depending on $\lambda$ and $d$. 

If $k = 0$, let $r$ be the lowest Hamming weight at which $f$ disagrees with $f(\vec 0)$, and redefine $k$ to be the number of inputs of Hamming weight $r$ which do not agree with $f(\vec 0)$. Then, essentially the same calculation shows that
\begin{align*}
\phi(f_*(P^{\otimes d}), f_*(Q^{\otimes d})) - \phi(P,Q) - \log{d}
= \log\left(rk(p^{r - 1} + \cdots + q^{r - 1})\right) - \lambda \log \frac{k p^r + k q^r}{2} + o(p+q).
\end{align*}
Since $k$ and $r$ are bounded in terms of the fixed quantity $d$, we see that the above expression is at most
$$C_1 \log(p+q) - \lambda \left(C_2 \log(p+q)\right) + O(1) $$
for some constants $C_1,C_2 > 0$ depending on $d$. 
Therefore as long as $\lambda,p,q$ are sufficiently small this quantity is bounded above by some $-c < 0$ independent of $p$ and $q$. 

Finally, if $p,q$ are bounded away from $0$ and $1$, then we know by Theorem~\ref{thm:restricted-sdpi-multibit} that 
$$\SKL(f_*(P^{\otimes d}, Q^{\otimes d})) \le \eta d \SKL(P,Q)$$ 
for some $\eta$ < 1, depending only on how close our region gets to the boundaries of the interval $[0,1]$. Furthermore, because $f$ is non-constant, we know that $\Pr_{P^{\otimes d}}(f = 1) \ge \min\{(1-p)^{d}, p^{d}\}$. Therefore if $\lambda$ is chosen sufficiently small (depending on $d$) with respect to $\eta$, we get
\[ \phi(f_*(P^{\otimes d}), f_*(Q^{\otimes d})) - \log {d} \le \phi(P,Q) - \eta/2. \]
Combining these bounds gives the result.
\end{proof}

\begin{theorem}\label{thm:no-1bit-reconstruction}
There exists $\nu_1$ such that $4d \nu_1^2 > 1$ (i.e. reconstruction is information-theoretically possible) but no 1-bit message passing scheme with $f$ constant on each level solves the reconstruction problem for any $\nu$ with $\nu \le \nu_1$.
\end{theorem}
\begin{proof}
If any of the reconstruction functions $f_t$ is a constant, then clearly reconstruction will fail, so we assume henceforth that all of the reconstruction functions $f_t$ are non-constant. Then, it follows from 
Lemma~\ref{lem:skl-bounded} and Theorem~\ref{thm:regularized-sdpi} that for one step of the dynamics
\begin{align*}
\phi(P^{+}_{n},P^{-}_{n}) 
&\leq \phi(P_{n-1}^{+}(\nu), P_{n-1}^{-}(\nu))+\log{d} - c \\
&= \log\SKL(P_{n-1}^{+}(\nu), P_{n-1}^{-}(\nu)) - \lambda \left(\log \frac{p_{n-1}^{+} + p_{n-1}^{-}}{2} + \log \left(1 - \frac{p_{n-1}^{+} + p_{n-1}^{-}}{2}\right)\right) + \log{d} - c \\
&\leq \log\SKL(P^{+}_{n-1},P^{-}_{n-1}) - \lambda \left(\log \frac{p_{n-1}^{+} + p_{n-1}^{-}}{2} + \log \left(1 - \frac{p_{n-1}^{+} + p_{n-1}^{-}}{2}\right)\right) + \log{4\nu^{2}} + \log{d} - c\\
&\leq L(P^{+}_{n-1},P^{-}_{n-1}) +  \log{4d\nu^{2}} - c, 
\end{align*}
so that
$$\phi(P^{+}_{n},P^{-}_{n}) - \phi(P^{+}_{n-1},P^{-}_{n-1}) \leq \log{4d\nu^{2}} - c < 0 $$
for all sufficiently small $\nu$ with $4 d \nu^{2} > 1$. Therefore, for such choices of $\nu$, $\phi(P_n^{+},P_n^{-})$ tends to $-\infty$ as $n \to \infty$. Since $$\lambda \left(\log \frac{p_n + q_n}{2} + \log\left(1 - \frac{p_n + q_n}{2}\right)\right)$$ is bounded above by an absolute constant, this implies that $\SKL(P_n,Q_n) \to 0$, which gives the desired conclusion. 
\end{proof}

\subsection{Complicated dynamics in the multibit setting}\label{apx:complicated}
The following example illustrates some of the difficult phenomena that appear when the alphabet has size at least $3$:
\begin{example}\label{example:cycling-dynamics}
Consider reconstruction with a 3-state alphabet $\Sigma = \{0,1,2\}$. We define a fixed reconstruction function $f$ inspired by ``intransitive preferences'': (1) when $x$ is supported on $\{0,1\}$, $f$ restricts to OR, i.e. $f(x) = 1$ unless $x = \vec{0}$; (2) when $x$ is supported on $\{1,2\}$, $f(x) = 2$ unless $x = \vec{1}$, and (3) when $x$ is supported on $\{2,0\}$, $f(x) = 0$ unless $x = \vec{2}$. Finally, when $x$ has full support $f$ equals the plurality with ties broken arbitrarily. Then, initially with $P_0,Q_0$ lying near the middle of the simplex, the discrete time dynamical system $(P_t,Q_t)_{t = 0}^{\infty}$ spirals (at a rapidly slowing rate) around the simplex, getting arbitrarily close to the boundary/corners of the simplex without ever hitting them. 
\end{example}

\section{Impossibility of multibit reconstruction near criticality}
\label{sec:appendix-multibit}
\subsection{Proof of Lemma~\ref{lem_approx}}
\label{sec:appendix-proof-linearization}
As mentioned before, Lemma~\ref{lem_approx} is a special
case of the following more general result:
\begin{lem}
\label{lem_approx_detailed}
 Suppose $a_u$, $b_v$ are functions on arbitrary sets $\mathcal{U}$, $\mathcal{V}$ and taking values in $[-1,1]$.
Suppose that $P_U$, $P_V$ are arbitrary distributions, under which $a_u$, $b_v$ have zero mean and 
\begin{align*}
 \mu=\E|a_U|^p
 =\E|b_V|^p
 \le \frac{1}{2^{1+p}}.
\end{align*}
Consider two random variables,
\begin{align*}
\bar{S}&=a_{\bar{U}}+b_{\bar{V}},\, P_{\bar{U}\bar{V}}(u,v)\sim P_U(u)P_V(v);
\\
 \hat{S}&=\frac{a_{\hat{U}}+b_{\hat{V}}}{1+a_{\hat{U}}b_{\hat{V}}},\, P_{\hat{U}\hat{V}}(u,v)\sim P_U(u)P_V(v)(1+a_ub_v).
\end{align*}
Then for any $p\ge 1$,
\begin{align*}
 W_p^p(\hat{S},\bar{S})\le \mu\alpha_p(\mu)
\end{align*}
where $\alpha_p(\cdot)$ is a function satisfying $\lim_{\mu\to0}\alpha_p(\mu)=0$.
Explicitly, 
\begin{align}
  \alpha_p(\mu):=
  \mu^{\frac{1}{1+p}}
\left(4+\frac{8}{(1-\mu^{\frac{1}{1+p}})^{\frac{1}{p}}}\right)^p.
\label{e_alpha}
\end{align}
\end{lem}

\begin{proof}
Consider any $\delta\in (0,0.5]$, and define
\begin{align*}
 \mathcal{T}:=\{(u,v)\colon |a_ub_v|\le \delta\}.
\end{align*}
Then by the Markov inequality,
\begin{align*}
 \mathbb{P}[(\bar{U},\bar{V})\in\mathcal{T}^c]
 &\le \delta^{-p}\E|a_{\bar{U}}b_{\bar{V}}|^p
 \\
 &\le \delta^{-p}\mu^2
\end{align*}
and 
\begin{align*}
  \mathbb{P}[(\hat{U},\hat{V})\in\mathcal{T}^c]
  \le
  2\mathbb{P}[(\bar{U},\bar{V})\in\mathcal{T}^c]
  \le 2 \delta^{-p}\mu^2.
\end{align*}
Define random variables
\begin{align*}
\bar{S}'&=\bar{S}1\{(\bar{U},\bar{V})\in\mathcal{T}\},
\\
 \hat{S}'&=\hat{S}1\{(\hat{U},\hat{V})\in\mathcal{T}\}.
\end{align*}
Then since $|\bar{S}|\le 2$ and $|\hat{S}|\le 1$ with probability one, we see that
\begin{align*}
W_p^p(\bar{S}',\bar{S})&\le 2^p\delta^{-p}\mu^2;
\\
 W_p^p(\hat{S}',\hat{S})&\le 2\delta^{-p}\mu^2,
\end{align*}
so the problem is reduced to bounding 
$W_p^p(\hat{S}',\bar{S}')$.
Define $S''=(a_{\hat{U}}+b_{\hat{V}})1\{(\hat{U},\hat{V})\in\mathcal{T}\}$.
Then using Proposition~\ref{prop_w} with $\tau=\frac{\delta}{1-\delta}$ we have
\begin{align*}
 W_p(S'',\bar{S}')
 &\le \left(\frac{\delta}{1-\delta}\right)^{1/p}((\E[|S''|^p])^{1/p}
 +(\E[|\bar{S}'|^p])^{1/p})
\\
&\le 6\left(\frac{\delta}{1-\delta}\right)^{1/p}\mu^{1/p}
 \end{align*}
 where the last step used the facts that
\begin{align*}
\E[|\bar{S}'|^p]
&\le \E[|\bar{S}|^p]
\\
&=\E[|a_{\bar{U}}+b_{\bar{V}}|^p]
\\
&\le 2^p\mu;
\\
\E[|S''|^p]
&\le \E[|a_{\hat{U}}+b_{\hat{V}}|^p] 
\\
&\le 2\E[|a_{\bar{U}}+b_{\bar{V}}|^p] 
\\
&\le 2^{p+1}\mu
\\
&\le 4^p\mu.
\end{align*}
 Moreover,
 \begin{align*}
 W_p^p(S'',\hat{S}')
 &\le \E\left[\left|\frac{1}{1+a_{\hat{U}}b_{\hat{V}}}-1\right|^p|a_{\hat{U}}+b_{\hat{V}}|^p1\{(\hat{U},\hat{V})\in\mathcal{T}\}\right]
 \\
  &\le
 \left(\frac{\delta}{1-\delta}\right)^p\E\left[|a_{\hat{U}}+b_{\hat{V}}|^p1\{(\hat{U},\hat{V})\in\mathcal{T}\}\right]
 \\
 &\le
  \left(\frac{\delta}{1-\delta}\right)^p\E\left[|a_{\hat{U}}+b_{\hat{V}}|^p\right]
  \\
  &\le 
  \left(\frac{2\delta}{1-\delta}\right)^p\mu.
\end{align*}
Finally applying the triangle inequality to the above bounds, we obtain
\begin{align*}
 W_p(\hat{S},\bar{S})/\mu^{1/p}
 &\le
  \mu^{1/p}\delta^{-1}(2+2^{1/p})
  +6\left(\frac{\delta}{1-\delta}\right)^{1/p}
  +\frac{2\delta}{1-\delta}
  \\
  &\le \frac{4\mu^{1/p}}{\delta}+8\left(\frac{\delta}{1-\delta}\right)^{1/p}.
 \end{align*}
 Choosing $\delta=\mu^{\frac{1}{1+p}}$ (which satisfies $\delta\le1/2$ since we assumed $\mu\le \frac{1}{2^{1+p}}$) gives
 \begin{align}
 W_p(\hat{S},\bar{S})/\mu^{1/p}\le 
  \mu^{\frac{1}{p(1+p)}}
  \left(4+\frac{8}{(1-\mu^{\frac{1}{1+p}})^{\frac{1}{p}}}\right).
  \label{e59}
\end{align}
\end{proof}

\begin{prop}\label{prop_w}
 Let $X\sim P$, $Y\sim Q$ be real-valued random variables.
 Let $\mu$ and $\nu$ respectively be the restrictions of $P$ and $Q$ on $\mathbb{R}\setminus\{0\}$ (that is, $\mu(\mathcal{A})=P(\mathcal{A}\setminus \{0\})$ for any $\mathcal{A}\subseteq\mathbb{R}$, and similarly for $\nu$).
Suppose that $\mu$ and $\nu$ are mutually absolutely continuous.
Let
\begin{align*}
\tau:=\max\left\{\left\|\frac{{\rm d}\mu}{{\rm d}\nu}-1\right\|_{\infty},\,
\left\|\frac{{\rm d}\nu}{{\rm d}\mu}-1\right\|_{\infty} \right\}.
\end{align*}
Then for any $p\ge 1$,
\begin{align*}
 W_p(X,Y)\le \tau^{1/p}((\E[|X|^p])^{1/p}+(\E[|Y|^p])^{1/p}).
\end{align*}
\end{prop}
\begin{proof}
 Let $Z\sim R$ where $R$ is the measure which equals $\mu\wedge \nu$ on $\mathbb{R}\setminus\{0\}$ and has a point mass at $\{0\}$. 
Here $\mu\wedge\nu$ is the measure having the property that for any $\mathcal{A}$,
\begin{align*}
  R(\mathcal{A})=\mu\left(\mathcal{A}\cap\left\{\frac{{\rm d}\mu}{{\rm d}\nu}\le1\right\}\right)
  +\nu\left(\mathcal{A}\cap\left\{\frac{{\rm d}\mu}{{\rm d}\nu}>1\right\}\right).
\end{align*}
A natural coupling between $X$ and $Z$ is the following:
let $B$ be independent of $X$ and uniformly distributed on $[0,1]$.
Set
\begin{align*}
Z=X1\left\{B\le\frac{{\rm d}(\mu\wedge\nu)}{{\rm d}\mu}(X)\right\}.
\end{align*}
Then 
\begin{align*}
W_p^p(X,Z)
&\le \E[|X-Z|^p]
\\
&\le
\E\left[|X|^p1\left\{B>\frac{{\rm d}(\mu\wedge\nu)}{{\rm d}\mu}(X)\right\}\right]
\\
&\le \E\left[|X|^p1\left\{B>1-\tau\right\}\right]
\\
&\le \tau\E[|X|^p].
\end{align*}
Note that $Z$ is equal in distribution to $Y1\left\{B\le\frac{{\rm d}(\mu\wedge\nu)}{{\rm d}\nu}(Y)\right\}$,
therefore similarly we have $W_p^p(Y,Z)\le \tau\E[|Y|^p]$, and the claim follows from the triangle inequality of the Wasserstein distance.
\end{proof}

\subsection{Proof of Proposition~\ref{prop_approx1}}\label{subsec:approx1}
\begin{proof}[Proof of Proposition~\ref{prop_approx1}]
It suffices to prove the result in the case where $L = \infty$. 
In this case the supremum in \eqref{e_xi} is attained by Belief Propagation
which induces a symmetric distribution on $S_n$ (for all $n$).
Therefore, from the upper bound in 
Lemma~\ref{lem_approx1} we know that as long as $\xi \ne 0$,
\begin{equation*} \frac{1}{2(1 - 2\varepsilon)^{2}} \le 1 - \xi/4, \end{equation*}
and rearranging gives the result.
\end{proof}

\begin{lem}\label{lem_approx1}
 Suppose $a_u$, $b_v$ are functions on arbitrary sets $\mathcal{U}$, $\mathcal{V}$ and taking values in $[-1,1]$.
Suppose that $P_U$, $P_V$ are arbitrary distributions, under which $a_u$, $b_v$ have zero mean and 
equal variance, i.e. $\E[a_U]^2 =\E[b_V]^2$.
Consider two random variables,
\begin{align*}
\bar{S}&=a_{\bar{U}}+b_{\bar{V}},\, P_{\bar{U}\bar{V}}(u,v)\sim P_U(u)P_V(v);
\\
 \hat{S}&=\frac{a_{\hat{U}}+b_{\hat{V}}}{1+a_{\hat{U}}b_{\hat{V}}},\, P_{\hat{U}\hat{V}}(u,v)\sim P_U(u)P_V(v)(1+a_ub_v).
\end{align*}
Finally, suppose that $a_U$ and $b_V$ are symmetric random variables, i.e., 
$a_U$ equals $-a_U$ in distribution where $U\sim P_U$, and similarly for $b_V$.
Then 
\begin{align}
1-\frac{1}{2}\E[\bar{S}^2]
\le \frac{\E[\hat{S}^2]}{\E[\bar{S}^2]}\le 1-\frac{1}{4}\E[\bar{S}^2]. 
\label{e_symmetric}
\end{align}
\end{lem}
\begin{proof}
Since 
\begin{align*}
\E[\hat{S}^2]
&=\E\left[\left(\frac{a_{\hat{U}}+b_{\hat{V}}}{1+a_{\hat{U}}b_{\hat{V}}}\right)^2\right]
\\
&=\E\left[\frac{\left(a_{\bar{U}}+b_{\bar{V}}\right)^2}{1+a_{\bar{U}}b_{\bar{V}}}\right],
\end{align*}
we can compute that 
\begin{align}
 \E[\bar{S}^2]-\E[\hat{S}^2]
=\E\left[\frac{(a_{\bar{U}}+b_{\bar{V}})^2
a_{\bar{U}}b_{\bar{V}}}
{1+a_{\bar{U}}b_{\bar{V}}}\right].
\label{e20}
\end{align}
Under the symmetric distribution assumption, 
we can replace $a_{\bar{U}}$ with $-a_{\bar{U}}$ in the above without changing the value of the expectation, resulting in
\begin{align}
 \E[\bar{S}^2]-\E[\hat{S}^2]
=\E\left[\frac{-(a_{\bar{U}}-b_{\bar{V}})^2
a_{\bar{U}}b_{\bar{V}}}
{1-a_{\bar{U}}b_{\bar{V}}}\right].
\label{e21}
\end{align}
Adding \eqref{e20} and \eqref{e21} and dividing by $2$,
we obtain
\begin{align}
 \E[\bar{S}^2]-\E[\hat{S}^2]
&=\E\left[\frac{
a_{\bar{U}}^2b_{\bar{V}}^2
(2-a_{\bar{U}}^2-b_{\bar{V}}^2)}
{1-a_{\bar{U}}^2b_{\bar{V}}^2}\right]
\label{e_ss}
\\
&\ge \E\left[
a_{\bar{U}}^2b_{\bar{V}}^2
\right]
\label{e22}
\\
&=\frac{\E^2[\bar{S}^2]}{4} \label{e23}
\end{align}
where \eqref{e22} follows from 
\begin{align*}
2-a_{\bar{U}}^2-b_{\bar{V}}^2
-(1-a_{\bar{U}}^2b_{\bar{V}}^2)
=(1-a_{\bar{U}}^2)(1-b_{\bar{V}}^2)\ge 0
\end{align*}
and \eqref{e23} from the assumption of equal variances. 
This proves the second inequality in \eqref{e_symmetric}.
To prove the first inequality in \eqref{e_symmetric},
note that
\begin{align*}
\frac{1-a_{\bar{U}}^2b_{\bar{V}}^2}{2-a_{\bar{U}}^2-b_{\bar{V}}^2}
&=1-\frac{(1-a_{U}^2)(1-b_{\bar{V}}^2)}{2-a_{\bar{U}}^2-b_{\bar{V}}^2}
\\
&\ge  1-\frac{(1-a_{U}^2)(1-b_{\bar{V}}^2)}{(1-a_{\bar{U}}^2)^2+(1-b_{\bar{V}}^2)^2}
\\
&\ge \frac{1}{2}
\end{align*}
and apply \eqref{e_ss}.
\end{proof}
\subsection{Proof of Lemma~\ref{lem_fourth}}
\label{sec:appendix-fourth-moment-proof}
\begin{proof}
Suppose $\varepsilon_1$ is sufficiently close to $\varepsilon_c$ so that
\begin{align}
\frac{2\lambda^4(\varepsilon_1)}{3}
+10000\alpha_4(\omega(\varepsilon_1))
\le \frac{3}{4}.
\label{e27}
\end{align}
where $\omega$ is defined in \eqref{e_omega}, $\alpha_{4}$ is defined in \eqref{e_alpha}, 
and $\lambda(\varepsilon_1):=\sqrt{2}(1-2\varepsilon_1)$.
We choose $h_1=h_1(\varepsilon,L)$ to be such that
 \begin{align}
 \xi_{h_1-1}(\varepsilon,L)\le 1.1\xi(\varepsilon,L),
 \label{e_h1}
 \end{align}
 and then put $h_2:=h_1+\log_{4/3}\frac{1}{\xi^2}$.
 Our proof strategy is to derive recursive relations for $\mu_n$.
 First, using the triangle inequality for the Wasserstein distance, we have
 \begin{align*}
 \hat{\mu}_n^{1/4}\le \bar{\mu}_n^{1/4}+W_4(\hat{S}_n,\bar{S}_n). 
 \end{align*}
Recall that $\bar{S}_n$ equals in distribution to $(1-2\varepsilon)(S_{n-1}+S_{n-1}')=\lambda\cdot \frac{S_{n-1}+S_{n-1}'}{\sqrt{2}}$ where $S_{n-1}'$ is an independent copy of $S_{n-1}$.
Thus from Lemma~\ref{lem_approx_detailed} we have
\begin{align}
 \hat{\mu}_n
 &\le
 \left(\bar{\mu}_n^{1/4}
 +\frac{\lambda\mu_{n-1}^{1/4}}{\sqrt{2}}
 \alpha_4^{1/4}\left(\frac{\lambda^4\mu_{n-1}}{4}\right)\right)^4
 \\
 &\le \frac{4}{3}\bar{\mu}_n+5000\lambda^4\mu_{n-1}
 \alpha_4\left(\frac{\lambda^4\mu_{n-1}}{4}\right)
 \label{e29}
 \\
 &\le \frac{3}{4}\mu_{n-1}
 +3\xi^2
 \label{e30}
\end{align}
where 
\begin{itemize}
\item \eqref{e29} used the elementary inequality $(x+y)^4\le \frac{4}{3}x^4+20000y^4$ for $(x,y)\in [0,\infty)^2$.
\item To see \eqref{e30}, 
we first use Proposition~\ref{prop_approx1}, 
and the assumption of $n\ge h_1$ to bound
\begin{align}
 \mu_{n-1}&= \E[S_{n-1}^4]
 \\
 &\le \E[S_{n-1}^2]
 \\
 &\le \xi_{n-1}
 \\
 &\le 2\omega(\varepsilon).
 \label{e31}
\end{align}
Then \eqref{e30} follows by expanding $\bar{\mu}_n=\frac{\lambda^4}{2}\mu_{n-1}+\frac{3\lambda^4}{2}\sigma_{n-1}^4$ and performing some basic calculations using the assumptions \eqref{e27} and \eqref{e_h1}.
\end{itemize}
Next, using $\E[\hat{S}|S_n]=S_n$ and Jensen's inequality we have
\begin{align}
 \mu_n
 &\le \hat{\mu}_n. 
 \label{e38}
\end{align}
Combining \eqref{e30} and \eqref{e38},
we obtain the following recursion which holds for any $n\ge h_1$:
\begin{align}
 \mu_n-12\xi^2\le \frac{3}{4}(\mu_{n-1}-12\xi^2).
 \label{e57}
\end{align}
Since $\mu_{h_1}\le 1$, by the definition of $h_2$ we have
\begin{align}
 \mu_n\le 13\xi^2,\quad\forall n\ge h_3.
 \label{e59}
\end{align}
\end{proof}

\subsection{Proof of Lemma~\ref{lem_clt}}
\label{sec:appendix-proof-clt}
\begin{proof}
First, let us specify the choice of $\varepsilon_2$.
Define 
\begin{align}
\eta=\eta(\varepsilon):=\lambda  \left(1+
 \alpha_2^{1/2}\left(\lambda^2\omega(\varepsilon)\right)\right),
\end{align}
(which is a function of $\varepsilon$ since $\lambda$ is a function of $\varepsilon$), where $\omega(\cdot)$ and $\alpha_2(\cdot)$ were defined in \eqref{e_omega} and Lemma~\ref{lem_approx}. 
Then put
\begin{align*}
\varepsilon_2:=\max\{\varepsilon_1,\,\eta^{-1}(2^{1/6}),\,
\omega^{-1}(1/5)\},
\end{align*}
where $\varepsilon_1$ was defined in Lemma~\ref{lem_fourth}. 

Consider any $t\in\{1,2,\dots\}$ (to be optimized later).
Let $h_3=h_3(\varepsilon,L,\delta)>h_2$ (where $h_2$ is from Lemma~\ref{lem_fourth}) be such that $\xi_h\le (1+\delta)\xi$ for any $h\ge h_3$.
Now consider any tree of height $h\ge h_3+t$,
and a reconstruction algorithm such that $\sigma_h^2\ge (1-\delta)\xi$.
To bound the non-Gaussianness of $S$,
\begin{align}
&\quad W_2(S_h,\sigma_{h-t}G)
\nonumber\\
&\le
\sum_{n=h-t+1}^h
W_2\left(\frac{S_n^{(1)}+\dots
+S_n^{(2^{h-n})}}{2^{\frac{h-n}{2}}}
,
\frac{S_{n-1}^{(1)}+\dots
+S_{n-1}^{(2^{h-n+1})}}{2^{\frac{h-n+1}{2}}}\right)
+W_2\left(\frac{S_{h-t}^{(1)}+\dots
+S_{h-t}^{(2^t)}}{2^{\frac{t}{2}}}
,\sigma_{h-t}G\right)
\\
&\le
\sum_{n=h-t+1}^h
W_2\left(S_n,
\frac{S_{n-1}+S_{n-1}'}{\sqrt{2}}\right)
+W_2\left(\frac{S_{h-t}^{(1)}+\dots
+S_{h-t}^{(2^t)}}{2^{\frac{t}{2}}}
,\sigma_{h-t}G\right)
\label{e40}
\end{align}
where we used the triangle inequality and the subadditivity of Wasserstein distance \cite[Proposition~7.17]{villani2003topics}, and $S_n^{(1)},S_n^{(2)},\dots$ denote i.i.d.\ copies of $S_n$.
But note that for each $n\in\{h-t+1,\dots,h\}$,
\begin{align}
W_2\left(S_n,\frac{S_{n-1}+S_{n-1}'}{\sqrt{2}}\right)
&\le 
W_2(S_n,\hat{S}_n)
+W_2(\hat{S}_n,\bar{S}_n)
+W_2\left(\bar{S}_n,\frac{S_{n-1}+S_{n-1}'}{\sqrt{2}}\right)
\\
&\le \sqrt{\hat{\sigma}_n^2-\sigma_n^2}
+\sqrt{\frac{\lambda^2\sigma_{n-1}^2}{2}\cdot\alpha_2\left(\frac{\lambda^2\sigma_{n-1}^2}{2}\right)}
+(\lambda-1)\sigma_{n-1}
\label{e42}
\end{align}
where we used $W_2\left(\bar{S}_n,\frac{S_{n-1}+S_{n-1}'}{\sqrt{2}}\right)\le W_2(\lambda S_{n-1},S_{n-1})\le (\lambda-1)\sigma_{n-1}$.
We next bound the sum of \eqref{e42} over $n$ by showing that $\sigma_n$ cannot increase too fast in $n$.
For any $n\in \{h-t+1,h\}$,
observe that
\begin{align}
 \hat{\sigma}_n
 &\le \bar{\sigma}_n
 +\frac{\bar{\sigma}_n}{\sqrt{2}}
 \alpha_2^{1/2}\left(\frac{\lambda^2\sigma_{n-1}^2}{2}\right)
 \\
 &\le \bar{\sigma}_n
 \left(1+
 \alpha_2^{1/2}\left(\lambda^2\omega(\varepsilon)\right)\right)
 \label{e51}
\end{align}
where \eqref{e51} follows from 
\begin{align}
\sigma_k^2\le \xi_k\le (1+\delta)\xi\le 2\xi\le 2 \omega(\varepsilon),
\quad\forall k=h-t,h-t+1,\dots,h.
\end{align}
using the upper bound of $\omega(\varepsilon)$ from Proposition~\ref{prop_approx1}.
We then have
\begin{align}
 \frac{\sigma_n}{\sigma_{n-1}}
 &\le
\frac{\hat{\sigma}_n}{\sigma_{n-1}}
\label{e61}
 \\
 &= 
 \frac{\hat{\sigma}_n}{\bar{\sigma}_n}
 \cdot  \frac{\bar{\sigma}_n}{\sigma_{n-1}}
 \\
 &= \frac{\lambda\hat{\sigma}_n}{\bar{\sigma}_n}
 \\
 &\le \eta.
 \label{e63}
 \end{align}
Now,
\begin{align}
\sum_{n=h-t+1}^h\sqrt{\hat{\sigma}_n^2-\sigma_n^2}
 &\le
t\sqrt{\frac{1}{t}
 \sum_{n=h-t+1}^h[\hat{\sigma}_n^2-\sigma_n^2]}
 \\
 &\le t\sqrt{\frac{1}{t}
 \sum_{n=h-t+1}^h[\hat{\sigma}_n^2-\sigma_{n-1}^2]+\frac{\sigma_{h-t}^2-\sigma_h^2}{t}}
 \\
 &\le
 t\sqrt{2(\eta^2-1)\xi+\frac{2\delta\xi}{t}}
 \\
 &\le 
 t\sqrt{2(\eta^2-1)\xi}+t\sqrt{2\delta\xi/t}
 \label{e80}
\end{align}
where \eqref{e80} follows since \eqref{e61}-\eqref{e63} shows $\hat{\sigma_n}\le \eta\sigma_{n-1}$.
Moreover,
\begin{align}
\sum_{n=h-t+1}^h 
\sqrt{\frac{\lambda^2\sigma_{n-1}^2}{2}\cdot\alpha_2\left(\frac{\lambda^2\sigma_{n-1}^2}{2}\right)}
&\le
t\sqrt{\lambda^2\xi\cdot\alpha_2\left(\lambda^2\omega(\varepsilon)\right)}
;
\end{align}
\begin{align}
\sum_{n=h-t+1}^h (\lambda-1)\sigma_{n-1}
\le t(\lambda-1)\sqrt{2\xi}.
\end{align}
Also, 
by Rio's CLT (see \cite{rio1998} \cite{rio2009} and also \cite[Theorem~1.1]{bobkov2018berry}) we have
\begin{align}
W_2\left(\frac{S_{h-t}^{(1)}+\dots
+S_{h-t}^{(2^t)}}{2^{\frac{t}{2}}}
,\sigma_{h-t}G\right)
&\le
\frac{c_2\sigma_{h-t}}{\sqrt{2^t}}\sqrt{\E\left[
\left(\frac{S_{h-t}}{\sigma_{h-t}}\right)^4\right]}
\\
&\le c_22^{-t/2}\cdot\frac{\sqrt{13\xi^2}}{\sigma_{h-t}}
\\
&\le c_2\eta^t2^{-t/2}\cdot\sqrt{\frac{13\xi}{1-\delta}}
\label{e66}
\\
&\le c_22^{-t/3}\sqrt{26\xi}
\label{e67}
\end{align}
where $c_2$ is some absolute constant.
\eqref{e66} used the fact that $\sigma_{h-t}\ge \eta^{-t}\sqrt{(1-\delta)\xi}$
which in turn follows from \eqref{e63}.
\eqref{e67} follows since the definition of $\varepsilon_2$ ensures that $\eta\le 2^{1/6}$.
Plugging these into \eqref{e42} and then \eqref{e40},
we obtain
\begin{align}
W_2(S_h,\sqrt{\xi}G)
&\le t\sqrt{2(\eta^2-1)\xi}+t\sqrt{2\delta\xi/t}
+
t\sqrt{\lambda^2\xi\cdot\alpha_2\left(\lambda^2\omega(\varepsilon)\right)}
+t(\lambda-1)\sqrt{2\xi}
+
c_22^{-t/3}\sqrt{26\xi}
\label{e_59}
\\
&\le \sqrt{\xi}\left(c_3(\varepsilon_c-\varepsilon)^{1/12}t
+\sqrt{2\delta t}
+c_22^{-t/3}\right)
\label{e70}
\end{align}
where $c_3$ denotes an absolute constant.
To see \eqref{e70}, note that $\alpha_2(\mu)=O(\mu^{1/3})$ for $\mu<1/2$.
The assumption that $\varepsilon>\varepsilon_2$ implies that 
$\alpha_2(\lambda^2\omega(\varepsilon))=O((\varepsilon_c-\varepsilon)^{1/3})$,
and hence $\eta=1+O((\varepsilon_c-\varepsilon)^{1/6})$.
Finally, choosing $t=\log\frac{1}{\varepsilon_c-\varepsilon}$ yields the desired result.
\end{proof}
\section{Proof of achievability}\label{sec:achievability}
In this section we prove the upper-bound on $L$ in Theorem~\ref{thm:main}.
We adopt the following algorithm, which recursively quantizes $\hat{S}_n$ in the natural way:
\begin{itemize}
\item At the initial reconstruction level $n=0$ (i.e. the leaves), we apply a binary symmetric channel to generate the message $Y_v$ from $X_v$ at each leaf $v$, so that $\sigma_0^2=\frac{2(\lambda-1)}{\lambda^3}$.
\item At each reconstruction level $n\ge1$, (see Figure~\ref{fig_1})
define the reconstruction function $f$ by quantizing so that for $Y_a \in \{1,\ldots,L\}$,
$\E[X_a|Y_{b,c}]\in\left[\frac{2(Y_a-1)}{L}-1,\frac{2Y_a}{L}-1\right]$ with probability $1$ and $\E[X_a|Y_a]$ is symmetric (in the sense that $\E[X_a|Y_a]$ and $-\E[X_a|Y_a]$ have the same distribution).
Note that $Y_a$ is essentially the index of the quantization interval for $\E[X_a|Y_{b,c}]$,
with the boundary case assigned in a way to ensure symmetry (for example, adopt the rule that the interval nearer to $0$ is selected when $\E[X_a|Y_{b,c}]$ is on the boundary between two quantization intervals).
\end{itemize}
Note that this recursive rule ensures that $S_n$ is symmetric.
Conditioned on any $S_n$,
we see that $\hat{S}_n$ is distributed on an interval of length $2/L$.
Thus ${\rm Var}(\hat{S}_n|S_n)\le 1/L^2$,
and 
\begin{align}
\sigma_n^2=\hat{\sigma}_n^2-{\rm Var}(\hat{S}_n|S_n) \ge \hat{\sigma}_n^2-1/L^2.
\label{e11}
\end{align}
Now Lemma~\ref{lem_achiev} given below provides conditions of $\varepsilon$ and $L$ under which $\xi(\varepsilon,L)\ge 2(\lambda-1)/\lambda^3$,
which completes the proof of the upper bound in Theorem~\ref{thm:main}.

\begin{lem}\label{lem_achiev}
 There exists $\varepsilon_3\in (0,\varepsilon_c)$ such that for any $\varepsilon\in (\varepsilon_3,\varepsilon_c)$, and $L\ge\frac{\lambda^3}{2(\lambda-1)^2}$,
the algorithm has the following property:
if $\sigma_{n-1}^2\in[A,B]$ for some $n\in\{1,2,\dots\}$, then the level-$n$ algorithm ensures that $\sigma_n^2\in[A,B]$,
where $A:=\frac{2(\lambda-1)}{\lambda^3}$ and $B:=\frac{4(\lambda^2-1)}{\lambda^4}$.
In particular, the initialization of the algorithm at level-$0$ ensures that $\sigma_n^2\in[A,B]$ for all $n\in\{0,1,\dots\}$.
\end{lem}
\begin{proof}
The proof follows from Lemma~\ref{lem_approx1} (see Appendix~\ref{subsec:approx1}).
Consider the following functions of $t\in[0,2]$,
which come from the lower and upper bounds on $\E[\hat{S}^2]$ appearing in Lemma~\ref{lem_approx1}:
\begin{align*}
\phi(t)&:=t(1-t/2),
\\
\psi(t)&:=t(1-t/4).
\end{align*} 
Note that we have chosen $A$ and $B$ to be fixed points of $t\mapsto \lambda^{-1}\phi(\lambda^2t)$ and $t\mapsto \psi(\lambda^2 t)$ respectively,
and $A<B$ holds when $\lambda-1$ is sufficiently small.
Also, note that both $\phi$ and $\psi$ are increasing on $[0,1]$. 
We can make $\lambda-1$ sufficiently small so that $\lambda^2 B\le 1$.
Then
\begin{align*}
\lambda^{-1}\phi(\lambda^2t)&\ge A,
\\
\psi(\lambda^2t)&\le B,
\end{align*}
for any $t\in [A,B]$.
Now if $\sigma_{n-1}^2\in [A,B]$, 
By \eqref{e11} we have
\begin{align}
\sigma_n^2
&=
\hat{\sigma}_n^2\cdot\frac{\sigma_n^2}{\hat{\sigma}_n^2}
\\
&\ge \hat{\sigma}_n^2\left(1-\frac{1}{\hat{\sigma}_n^2L^2}\right)
\label{e17}
\end{align}
Note that the assumption on $L$ guarantees that $1-\frac{1}{\hat{\sigma}_n^2L^2}\ge 1-\frac{1}{\phi(\lambda^2A)L^2}\ge \frac{1}{\lambda}$.
We can therefore continue \eqref{e17} as 
\begin{align*}
\sigma_n^2\ge \phi(\lambda^2\sigma_{n-1}^2)\cdot \lambda^{-1}
\ge A.
\end{align*}
Moreover,
\begin{align*}
\sigma_n^2\le \hat{\sigma}_n^2\le \psi (\lambda^2\sigma_{n-1}^2)\le B.
\end{align*}
\end{proof}

\begin{remark}[Adaptivity to unknown $\epsilon$]
The algorithm we have described (or even the standard BP without any memory constraint), requires as input the noise parameter $\varepsilon$. However, this is not an essential feature: for any $\varepsilon$ bounded away from the threshold, if we just want an algorithm using a finite number of bits, a multilevel version of recursive majority with sufficiently many bits will achieve a nontrivial reconstruction guarantee \cite{mossel1998recursive}. The downside to this method is that the number of bits needed by the multilevel recursive majority algorithm has a suboptimal quantitative dependence on $\varepsilon$; we expect that a variant of the above quantized BP algorithm which also estimates $\varepsilon$ by looking at the noise in the bottom layers of the tree should be able to achieve the best of both worlds. 
\end{remark}

\section{Memory limit implies non-Gaussianness}
\label{sec:appendix-proof-non-gaussian-bound}
Recall that we defined non-Gaussianness in \eqref{e_ng}.
We give a short information-theoretic proof to the following basic result, 
which shows that bounded cardinality (or more generally, bounded entropy) implies that the non-Gaussianness is bounded below.
\begin{lem}\label{lem_quant}
Suppose that $Z$ is a random variable with unit variance.
Then
\begin{align*}
\mathcal{E}(Z)\ge \frac{1}{2\exp(H(Z))}.
\end{align*}
\end{lem}
\begin{proof}
We can assume without loss of generality that $\E[Z]=0$.
Suppose that the claim is not true, 
then $W_2(Z,\sigma G)< \frac{1}{2\exp(H(Z))}$ for some $\sigma$.
By the triangle inequality of the Wasserstein distance we have $\sigma\ge 1-\frac{1}{2\exp(H(Z))}$,
and thus
\begin{align*}
h(\sigma G)\ge \frac{1}{2}\log 2\pi e +\log\left(1-\frac{1}{2\exp(H(Z))}\right)
\end{align*}
Moreover, we can find a coupling such that $\E[(Z-\sigma G)^2]\le \frac{1}{4\exp(2H(Z))}$ (that is, the infimum in the definition of the Wasserstein distance should be achievable).
Then
\begin{align}
h(\sigma G|Z)
&=h(\sigma G-Z|Z)
\\
&\le \E\left[\frac{1}{2}\log\left( 2\pi e\E[(\sigma G-Z)^2|Z]\right)\right]
\label{e_maxent}
\\
&\le \frac{1}{2}\log\left( 2\pi e\E[(\sigma G-Z)^2]\right)
\\
&\le \frac{1}{2}\log (2\pi e)-\log 2 -H(Z)
\end{align}
where $h(\cdot|\cdot)$ denotes conditional differential entropy,
and \eqref{e_maxent} uses the fact that under a second moment constraint, the Gaussian distribution maximizes the differential entropy (see e.g. \cite{cover2012elements}).
Then
\begin{align}
H(Z) \ge I(\sigma G;Z)=h(\sigma G)-h(\sigma G|Z)
\ge \log\left(2-\frac{1}{\exp(H(Z))}\right)+H(Z).
\label{e_contra}
\end{align}
If $H(Z)=0$, then $Z$ is a constant and the claim is obviously true; otherwise, \eqref{e_contra} results in a contradiction.
Thus the claim is established.
\end{proof}


\section{BP distributional fixpoint analysis}\label{sec:bp-fixpoint}
\begin{lemma}\label{lem:bp-fixpoint-exists}
Fix broadcasting parameter $\varepsilon > 0$. Let $\rho$ denote the root of a $d$-ary tree of depth $n$, $V_n$ denote the set of leaves, and let $\mathcal{P}_n$ denote the distribution of the random variable $Y_n := \E[X_{\rho} | X_{V_n}]$ under the broadcast process. Then there exists a unique limiting distribution $\mathcal{P}$ such that $\mathcal{P}_n \to \mathcal{P}$ in distribution.
\end{lemma}
\begin{proof}
We take a natural coupling of all random variables by considering them to live on the infinite $d$-ary tree with root node $\rho$, where $V_n$ is just the set of nodes at depth $n$ on this tree. Observe that 
\[ Y_n = \E[X_{\rho} | X_{V_n}] = \E[\E[X_{\rho} | X_{V_n \cup L_{n - 1}}] | X_{V_n}] = \E[Y_{n - 1} | X_{V_n}] \]
under this coupling, because by the Markov property $X_{V_n}$ is conditionally independent of $X_{\rho}$ given $X_{L_{n - 1}}$. Also observe that $|Y_n| \le 1$ a.s. so the mgf exists.

Observe by Jensen's inequality for conditional expectation we have the following
inequality of mgfs:
\[ \phi_n(s) := \E[e^{s Y_n}] = \E[e^{s \E[Y_{n - 1} | X_{V_n}]}] \le \E[e^{s Y_{n - 1}}] \le \phi_{n - 1}(s). \]
Therefore by monotonicity, there is a limiting function $\phi$ such that $\phi_n \to \phi$ pointwise. By classical results in probability theory \citep{billingsley2013convergence} this implies the corresponding convergence in distribution result.
\end{proof}
\begin{remark}
If reconstruction is possible, then $\mathcal{P}$ is non-Gaussian (it is a symmetric mixture distribution of the laws given the root is $+$ and given the root is $-$, and these conditional laws must have non-zero TV).
\end{remark}
\begin{remark}
The above result shows that far up from the leaves, the marginal distribution of BP messages converges.
This convergence result does not hold for general message passing schemes. For example, if the message passing scheme treats even and odd depths differently (e.g. alternating between AND and OR) then it's easy for the analogous convergence result to fail.
\end{remark}
\begin{proof}[Proof of Corollary~\ref{corr:bp-fixpoint}]
This follows from Lemma~\ref{lem:bp-fixpoint-exists} and Lemma~\ref{lem_clt}, noting that when $L = \infty$ (no memory constraint) BP always achieves the supremum in \eqref{e_xi}. 
\end{proof}

\end{document}